\documentclass[11pt,letterpaper,onecolumn,accepted=2026-11-25]{quantumarticle}
\pdfoutput=1

\usepackage[utf8]{inputenc}
\usepackage[english]{babel}
\usepackage{csquotes}

\usepackage{graphicx}

\usepackage{amsmath, amssymb, amsfonts}
\usepackage{amsthm, mathrsfs, mathtools}

\usepackage{hyperref}
\usepackage[capitalise]{cleveref}

\usepackage[backend=biber, style=alphabetic,sorting=nyt,sortcites=true,maxbibnames=99]{biblatex}
\usepackage{enumitem} 
\usepackage{booktabs} 
\usepackage{microtype} 
\usepackage{dsfont} 
\usepackage{faktor} 

\let\epsilon\varepsilon

\let\phi\varphi

\theoremstyle{plain}
\newtheorem{theorem}{Theorem}[section]

\newtheorem{corollary}[theorem]{Corollary}
\newtheorem{proposition}[theorem]{Proposition}

\newtheorem{definition}[theorem]{Definition}
\newtheorem{example}[theorem]{Example}
\newtheorem{remark}[theorem]{Remark}
\newtheorem{maintheorem}{Theorem}

\newcommand*{\Tr}[1]{\mathop{}\!\mathrm{Tr}{\left(#1\right)}}

\def\id{\mathds{1}}
\newcommand{\poly}{\mathrm{poly}}

\newcommand{\ket}[1]{\ensuremath{\left|#1\right\rangle}}

\newcommand{\sandwich}[3]{\langle #1|#2 |#3\rangle  }

\newcommand{\cH}{\mathcal{H}}

\newcommand{\cA}{\mathcal{A}}
\newcommand{\cB}{\mathcal{B}}
\newcommand{\negl}{\mathrm{negl}}

\newcommand{\comp}{\mathrm{comp}}
\newcommand{\seq}{\mathrm{seq}}

\let\B\relax 
 
\newcommand{\A}{A_{a|x}}
\newcommand{\B}{B_{b|y}}
\newcommand{\hA}{\hat{A}_{a|x}}

\newcommand{\proofstep}[2]{%
  \par
  \addvspace{\medskipamount}%
  \noindent\emph{Step #1: #2}\par\nobreak
  \addvspace{\smallskipamount}%
}

\newcommand{\cG}{\mathcal{G}} 
\newcommand{\gamevalueqcopt}[1]{\omega_{\mathrm{qc}}(#1)} 
\newcommand{\gamevalueNPA}[2]{\omega_{\mathrm{NPA}}^{#2}(#1)} 
\newcommand{\gamevalueSeqNPA}[2]{\omega_{\mathrm{seqNPA}}^{#2}(#1)} 
\newcommand{\gamevalueCompile}[2]{\omega_{#1}({#2})} 
\newcommand{\gamevalueModNPA}[2]{\omega_{\mathrm{modNPA}}^{#2}(#1)} 
\newcommand{\gamevalueSparse}[2]{\omega_{\mathrm{sparse}}^{#2}(#1)} 
\newcommand{\etaLem}{\eta^{L}}
\newcommand{\etaThm}{\eta} 
\newcommand{\coeffnorm}[1]{\left\lVert #1\right\rVert_{\mathrm{coef}}}

\usepackage{xcolor}

\hypersetup{
  pdftitle    = {Quantitative Quantum Soundness for Bipartite Compiled Bell Games via the Sequential NPA Hierarchy},
  pdfauthor   = {Igor Klep; Connor Paddock; Marc-Olivier Renou; Simon Schmidt; Lucas Tendick; Xiangling Xu; Yuming Zhao},
  pdfsubject  = {Quantum information theory; cryptography},
  pdfkeywords = {compiled Bell games, sequential NPA hierarchy, quantitative quantum soundness, semidefinite programming SDP, sparse sums-of-squares SOS, MIPco = coRE, bipartite nonlocal games, quantum homomorphic encryption, weakly non-signaling decomposition},
  pdfcreator  = {LaTeX with hyperref},
  pdflang     = {en},
  pdfdisplaydoctitle = true
}

\title{Quantitative Quantum Soundness for Bipartite Compiled Bell Games via the Sequential NPA Hierarchy}

\author{Igor Klep}
\affiliation{Faculty of Mathematics and Physics, University of Ljubljana}
\affiliation{Faculty of Mathematics, Natural Sciences and Information Technologies, University of Primorska}
\affiliation{Institute of Mathematics, Physics and Mechanics, Ljubljana, Slovenia}

\author{Connor Paddock}
\affiliation{Department of Mathematics and Statistics, University of Ottawa, Canada}

\author{Marc-Olivier Renou}
\affiliation{Inria, CPHT, LIX, CNRS, \'Ecole Polytechnique, Institut Polytechnique de Paris, Palaiseau, France}

\author{Simon Schmidt}
\affiliation{Faculty of Computer Science, Ruhr University Bochum, Germany}

\author{Lucas Tendick}
\affiliation{Inria, CPHT, LIX, CNRS, \'Ecole Polytechnique, Institut Polytechnique de Paris, Palaiseau, France}

\author{Xiangling Xu}
\email{xu.xiangling@inria.fr}
\affiliation{Inria, CPHT, LIX, CNRS, \'Ecole Polytechnique, Institut Polytechnique de Paris, Palaiseau, France}

\author{Yuming Zhao}
\affiliation{QMATH, Department of Mathematical Sciences, University of Copenhagen, Denmark}

\begin{document}

\maketitle

\begin{abstract}
Compiling Bell games under cryptographic assumptions replaces the need for physical separation, allowing nonlocality to be probed with a single untrusted device.
While Kalai et al. (STOC'23) showed that this compilation preserves quantum advantages, its quantitative quantum soundness has remained an open problem.
We address this gap with two primary contributions.
First, we establish the \emph{first} quantitative quantum soundness bounds for all bipartite compiled Bell games via a newly formalized convergent \emph{sequential Navascu\'es-Pironio-Ac\'in (NPA) hierarchy}: any polynomial-time prover's score is controlled by a finite-level hierarchy value, and finite-level convergence gives a negligible gap to the commuting quantum value, or to the tensor-product quantum value under flat optimality.
Second, we develop the sequential NPA hierarchy as a computational tool of independent interest, establishing its convergence and a stopping criterion based on flatness, and identifying its conic dual as a sparse SOS hierarchy.
\end{abstract}

\section{Introduction}

Since Bell's groundbreaking work~\cite{bell1964einstein}, understanding and utilizing quantum nonlocality has been pivotal for both the conceptual foundations and practical applications of quantum theory.
A central tool for probing nonlocality is the study of correlations arising from (nonlocal) Bell games~\cite{brunner2014bell}, wherein multiple provers (also called players) coordinate their responses to questions chosen by a verifier (also called the referee). 
Quantum theory famously allows for correlations outside of classical theories, enabling quantum provers to sometimes achieve higher winning probabilities or ``higher scores'' than their classical counterparts.

The standard Bell game setup involves multiple, spatially separated provers who cannot communicate during the game, see \cref{fig:NonlocalCompiledBellGame}.(a).
This spatial separation is the typical way to enforce no-signaling constraints on the players or devices.
However, verifying spatial separations between multiple untrusted quantum devices can be practically challenging.
Moreover, from a theoretical standpoint, it is compelling to explore whether the power of quantum nonlocality can be verified and utilized using a single (untrusted) quantum device.
A naive attempt to adapt a bipartite (Alice and Bob) Bell game to a single prover is as follows: the single prover receives Alice's question $x$, computes her answer $a$, they subsequently receive Bob's question $y$ and compute his answer $b$. However, here the prover has full information about Alice's question (and answer) when deciding how to answer Bob's question.
This allows for coordination not permitted in the nonlocal case, and completely undermines the game's no-communication assumption.
To simulate the intended separation within a single device, the verifier must restrict information flow between the ``Alice'' and ``Bob'' rounds.

Homomorphic encryption (HE) offers a natural solution: the verifier can first encrypt Alice's question into $\mathrm{Enc}_{sk}(x)$ using a secret key $sk$, and ask the prover to provide an encrypted answer $\mathrm{Enc}_{sk}(a)$. In HE the prover does not know the secret key, and therefore never has a decryption of $\mathrm{Enc}_{sk}(x)$ in their possession. Nonetheless, the HE satisfies a \emph{correctness} functionality that enables the prover to compute an outcome $\alpha=\mathrm{Enc}_{sk}(a)$ as if they knew $x$, despite never being given $x$ in the plain (i.e., never given a decrypted $x$). The result is that, when Bob goes to make his computation based on $y$, it can no longer depend on $(x, a)$ in any meaningful way, as he only has access to their encryptions (\cref{fig:NonlocalCompiledBellGame}.(c)). 
However, to allow for quantum strategies, conventional HE will not suffice, because the player strategies involve quantum computations and entanglement: the HE of Alice's part of the strategy should not destroy her pre-shared entangled state with Bob. Therefore, we require a flavour of ``quantum'' HE which allows for the homomorphic evaluation of quantum circuits and satisfies a \emph{correctness with respect to auxiliary entangled systems} functionality. 
Fortunately, constructions of quantum homomorphic encryption (QHE) schemes for polynomial size quantum circuits, with these additional properties, were established in~\cite{brakerski2018quantum,mahadev2020classical}, based on the (post-quantum) security of the learning with errors (LWE) problem.
This approach was used by Kalai et al.~\cite{kalai2023quantum}, establishing the first \emph{compiled Bell games}, where a multipartite Bell game can be transformed into an interactive protocol with a single quantum prover using a QHE scheme, at the cost of involving a number of rounds proportional to the number of parties.
They demonstrated the \emph{classical soundness} of such compilation, meaning that a cheating classical prover cannot exceed the classical score at the standard Bell game.
Yet, an important issue was left open by their work: the \emph{quantum soundness}, that is whether the compilation preserves the maximal quantum score.

More explicitly, the possibility of a dishonest quantum prover achieving scores for the compiled Bell game that significantly exceeded what was possible in the spatially separated Bell game was not ruled out. To date, this issue has been resolved in the negative for a number of cases like XOR and other simple Bell inequalities~\cite{natarajan2023bounding,cui2024computational, baroni2024quantum, mehta2024self}, such as the CHSH game~\cite{clauser1969proposed}. For these games, it was shown that efficient quantum prover could only attain a winning probability negligibly (with respect to the encryption scheme's security parameter $\lambda$) greater than the quantum value of the original game.
Recently, some of us proved quantum soundness of all bipartite Bell games \emph{in the asymptotic limit of the security parameter going to infinity}~\cite{kulpe2024bound}. More precisely, we showed that for asymptotically large enough security parameter $\lambda$, the maximal quantum score at the compiled and standard Bell games is the same. 
Yet, this result is not quantitative, as it does not establish that the possible excess of the compiled score above the commuting quantum score is negligible in $\lambda$.

In this work, we obtain quantitative quantum soundness bounds for bipartite compiled Bell games via finite-level certificates, generalizing the results from~\cite{natarajan2023bounding,cui2024computational, baroni2024quantum, mehta2024self, kulpe2024bound}.
More precisely, we show that the score a dishonest quantum prover can achieve at the compiled Bell game can explicitly be upper-bounded by a \emph{sequential variant of the Navascu\'es-Pironio-Ac\'in (NPA) hierarchy}~\cite{navascues2008convergent, pironio2010convergent}, which we also fully characterize in this work.
When this hierarchy converges at a finite level, the compiled score exceeds the game's commuting quantum value by at most a negligible amount; under the stronger flat-optimality certificate, the same conclusion holds with respect to the tensor-product quantum value.
Note that the behavior of the sequential NPA hierarchy is a game-dependent property and independent of the compilation.
With our result, the verifier can be certain that the compiled score of every QPT strategy is bounded by a finite-level sequential NPA value plus a strategy-dependent negligible error.

\begin{figure}
    \centering
    \includegraphics[width=0.80\linewidth]{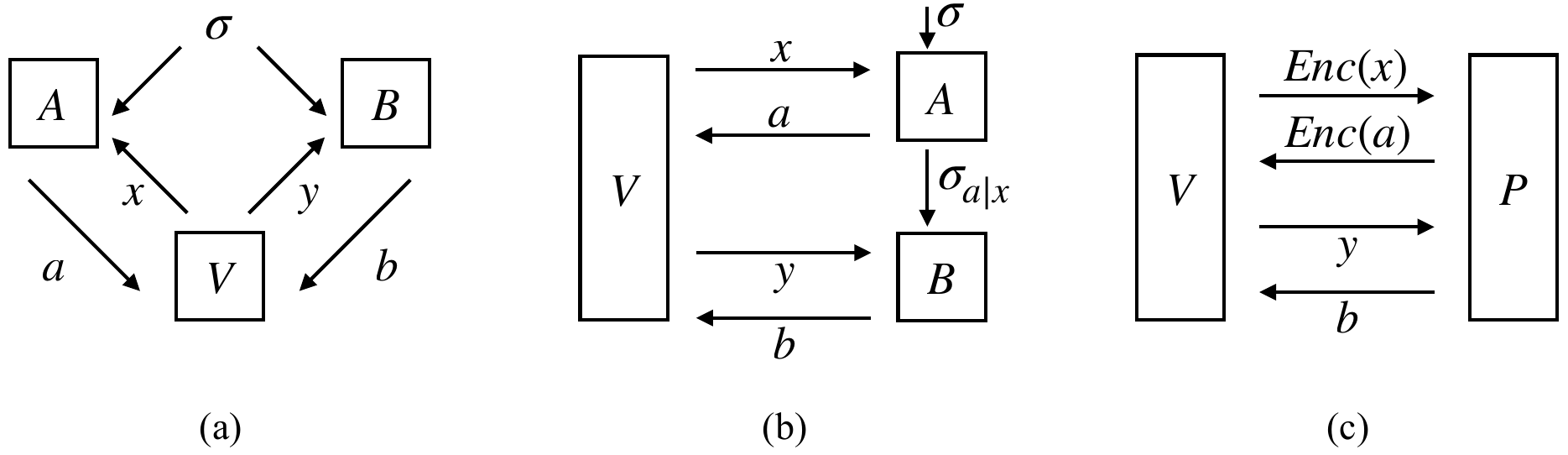}
    \caption{(a) \emph{Standard nonlocal Bell game}: A verifier $V$ sends questions $x$ and $y$ to two spatially separated provers, Alice $A$ and Bob $B$, who reply with answers $a$ and $b$. For example, using quantum theory, the provers may pre-share a quantum state $\sigma$ to generate better answers.
    (b) \emph{Sequential Bell game}: The verifier $V$ first sends question $x$ to Alice and receive her answer $a$; subsequently, $V$ sends $y$ to Bob who replies with $b$. This protocol is said to satisfy the strongly no-signaling condition if $B$'s response is independent of $A$'s question $x$. In the quantum realization, $A$ first receives a state $\sigma$, measures and produce a post-measurement state $\sigma_{a|x}$, and then forwards it to $B$, with condition $\sum_a \sigma_{a|x} = \sum_a \sigma_{a|x'}$ for all $x, x'$.
    (c) \emph{Compiled Bell game}: The verifier $V$ interacts with a single prover $P$. The verifier first sends an encrypted question Enc$(x)$ and receives an encrypted answer Enc$(a)$, while the second message pair $y$ and $b$ is transmitted unencrypted. The QHE scheme chosen by the verifier enforces a form of computational no-signaling.}
    \label{fig:NonlocalCompiledBellGame}
\end{figure}

\subsection{Nonlocal and compiled Bell games}
\noindent\emph{Nonlocal Bell games.}
In nonlocal Bell games, a verifier interacts with multiple spatially separated provers, who are unable to communicate during the game. The provers receive questions from and provide answers to the verifier according to a pre-agreed protocol.
The players win or lose based on a preset rule (see \cref{fig:NonlocalCompiledBellGame}.(a)) determined by a winning function or predicate.
The strategies that the provers adopt can be based on different resources available (e.g., classical or quantum), and the distinction between these theories is reflected in the corresponding \emph{Bell scores}.
The Bell score is the maximum winning probability using strategies permitted in a given resource or paradigm.
For a given Bell game $\cG$, we write $\gamevalueqcopt{\cG}$ its optimal commuting quantum score and $\omega_{\mathrm{q}}(\cG)$ its optimal tensor product quantum score\footnote{While the two are equivalent in finite dimensions, they are not the same in infinite-dimensions, and in fact there is a Bell game for which the scores are distinct~\cite{ji2021mip}.}.
More typically, the scores are compared in the classical and quantum cases. For example, in the CHSH game, the best classical Bell score is $0.75$, while the optimal quantum score is $\gamevalueqcopt{\cG_{\text{CHSH}}} = \omega_{\mathrm{q}}(\cG_{\text{CHSH}})=\cos^2(\pi/8)\approx 0.85$. This is often known as a game exhibiting quantum advantage.

\vspace{0.5em}
\noindent\emph{Compiled Bell games.}
To transform from multi-prover to a setup with a single-prover, the authors of~\cite{kalai2023quantum} introduce compiled Bell games $\cG_\comp$, in which the no-communication constraint between the provers is replaced by a cryptographic one, using a QHE scheme.
The QHE scheme used by the verifier is parameterized by a security parameter $\lambda$.
For a chosen $\lambda$, the scheme is secure against $\poly(\lambda)$-runtime attacks from the prover; i.e., the advantage in distinguishing encryption of different messages is bounded by a negligible function of $\lambda$.
The verifier in the compiled game $\cG_\comp$ sends an encrypted question $x$ and receives the encrypted answer $a$ back from the prover. 
The verifier then sends $y$ and receives $b$ (see \cref{fig:NonlocalCompiledBellGame}.(c)). 
Encryption is not required in the second round because the information is of no use later in the game.
In this setting, the prover's strategies for the compiled game are characterized by quantum polynomial time (QPT) circuits, denoted $S=(S_\lambda)_\lambda$, which upon obtaining $x$ produce the outcome $a$.
The winning probability of employing strategy $S$ in the game $\cG_{\comp}$ (with security parameter $\lambda$) is the compiled Bell score $\gamevalueCompile{\lambda}{\cG_{\comp}, S}$.
This compilation procedure guarantees classical \emph{soundness}. That is, no dishonest classical prover can exceed the maximal classical winning probability in the no communication setting. Furthermore, by the features of the QHE scheme, quantum \emph{completeness} is also guaranteed. That is, an honest quantum prover can achieve the optimal quantum score in the nonlocal case~\cite{kalai2023quantum}.

Establishing quantum soundness of $\gamevalueCompile{\lambda}{\cG_{\comp}, S}$ (i.e., that no dishonest quantum prover can exceed the maximal quantum score more than some quantitatively negligible function in $\lambda$) remains open. 
Recently, operator-algebraic techniques~\cite{kulpe2024bound} provided qualitative insights into this quantum compiled value in the asymptotic limit of the security parameter ($\lambda \to \infty$).
Their approach uses the fact that in the limit, compiled strategies correspond to strategies for \emph{sequential Bell games} satisfying the \emph{strongly non-signaling property} (see \cref{fig:NonlocalCompiledBellGame}.(b)).
These quantum strategies for sequential Bell games turn out to be equivalent to commuting quantum strategies~\cite{hughston1993complete,kulpe2024bound}, and so for every QPT strategy family $S=(S_\lambda)_\lambda$, we have $\limsup_{\lambda\to\infty}\gamevalueCompile{\lambda}{\cG_{\comp},S}\leq\gamevalueqcopt{\cG}$.

Yet, this asymptotic result gives no bound on how quickly the possible excess of $\gamevalueCompile{\lambda}{\cG_{\comp},S}$ above $\gamevalueqcopt{\cG}$ vanishes as $\lambda$ increases.
In particular, it does not rule out an excess that decays only polynomially in $\lambda$, or even more slowly.
Therefore, the quantitative quantum soundness of all compiled Bell games as proposed in~\cite{kulpe2024bound} remains an open problem. This is the main challenge that our work addresses.

\subsection{Main results}
Our work has two primary contributions.
(i) We give the first \emph{quantitative quantum soundness bounds} for bipartite compiled Bell games admitting finite-level sequential NPA convergence, showing that the compiled score is provably close to the game's ideal quantum score.
More generally, for all bipartite compiled Bell games, we obtain upper bounds for the compiled scores in terms of the sequential NPA hierarchy.
(ii) We formalize and fully characterize a sequential variant of the NPA hierarchy, a tool that underpins our analysis and is of independent interest.
In the following, we give more details.

\vspace{1.0em}
\noindent
\textit{Quantitative bound for bipartite compiled Bell scores.}
Let $\cG$ be any bipartite Bell game and $\cG_{\comp}$ its compiled version.
Our first main result upper-bounds the score $\gamevalueCompile{\lambda}{\cG_{\comp}, S}$ achievable by any QPT strategy~$S=(S_\lambda)_\lambda$ as the ideal commuting-operator score $\omega_{\mathrm{qc}}(\cG)$ plus two error terms: an approximation term $\epsilon(n)$ arising from level~$n$ of the sequential NPA hierarchy and a negligible cryptographic term (from the QHE scheme and the implementation of~$S$).
When the sequential NPA hierarchy attains the commuting score at some finite level $n_0$, or even admits a flat optimal solution at level $n_0$, then $\epsilon(n_0)=0$ and we obtain a negligible gap to the quantum value.
The precise statement is as follows.
\begin{maintheorem}[\cref{thm:GeneralQuantBoundFiniteSecurityQCValue,cor:MainBoundFiniteSecurityQCValue}]\label{thm:BoundScoreCompiledInformal}
    Consider any bipartite Bell game $\cG$ with commuting quantum score $\gamevalueqcopt{\cG}$.
    Then, for any QPT strategy $S=(S_\lambda)_\lambda$ and for every fixed integer $n \geq 1$, its achievable score $\gamevalueCompile{\lambda}{\cG_{\comp}, S}$ is bounded as
    \begin{align}
        \gamevalueCompile{\lambda}{\cG_{\comp}, S} \leq \gamevalueqcopt{\cG} + \epsilon(n) + \negl_{S,n}(\lambda),
    \end{align}
    where $\epsilon(n):=\gamevalueSeqNPA{\cG}{n}-\gamevalueqcopt{\cG}$ is the approximation error from the $n$-th level of the sequential NPA hierarchy, which monotonically vanishes as $n \to \infty$.
    The term $\negl_{S,n}(\lambda)$ is a negligible function (dependent on the QHE scheme, strategy $S$, and level $n$) that goes to zero faster than the reciprocal of any polynomial in $\lambda$.

    Furthermore, if the sequential hierarchy for $\cG$ converges at some finite level $n_0$, i.e., $\epsilon(n_0) = 0$, then
    \begin{align}
        \gamevalueCompile{\lambda}{\cG_{\comp}, S} \leq \omega_{\mathrm{qc}}(\cG) + \negl_S(\lambda),
    \end{align}
    where $\negl_S(\lambda)$ is some negligible function depending on the QHE scheme, the strategy $S$, and the finite level $n_0$.
    In particular, if the hierarchy admits a flat optimal solution at some level $n_0$, then
    \begin{align}
        \gamevalueCompile{\lambda}{\cG_{\comp}, S} \leq \omega_{\mathrm{q}}(\cG) + \negl_S(\lambda),
    \end{align}
    where $\omega_{\mathrm{q}}(\cG)$ is the optimal tensor product quantum score.
\end{maintheorem}

As in earlier works on compiled games, the negligible error is allowed to depend on the prover's QPT strategy; see, in particular,~\cite[Theorem~3.4]{cui2024computational} and~\cite[Proposition~4.6]{kulpe2024bound}.
More precisely, for a fixed QHE scheme satisfying the security assumption, every QPT strategy family $S=(S_\lambda)_\lambda$ and every fixed hierarchy level $n$ admit a negligible function $\negl_{S,n}$ for which the bound holds for all $\lambda$.
This includes arbitrary dishonest QPT strategies and does not require the verifier to know $S$.
On the other hand, our result does not specify a numerical value for the negligible error at a given security parameter $\lambda$.
To obtain such a value from our proof, one would need an explicit bound on how well an attacker with given computational resources can distinguish messages encrypted by the chosen QHE scheme, and then track this bound through the proof.

By first letting $\lambda \to \infty$ for fixed $n$, and then $n \to \infty$, we recover the asymptotic quantum soundness result of~\cite{kulpe2024bound}.
In addition, whenever the hierarchy is known to converge at a finite level, the second inequality establishes quantitative soundness with respect to the commuting quantum score.
Under the stronger flat-optimality assumption, this becomes quantitative quantum soundness with respect to the tensor product quantum score, generalizing~\cite{natarajan2023bounding,cui2024computational, baroni2024quantum, mehta2024self}.

The premises of finite-level convergence and flat optimality are intrinsic properties of the Bell game $\cG$ and its sequential NPA hierarchy, rather than assumptions on the compilation procedure.
For example, \cref{ex:CHSHFlatSequentialNPA} shows that the sequential NPA hierarchy for CHSH is exact at level $2$ and admits a flat optimal solution at level $3$.
The same argument extends to all two-player XOR games, with exactness at level $2$ and a flat optimal solution at some finite level.
Moreover, in the flat-optimal case, \cref{thm:StoppingCrteriaInformal} extracts a finite-dimensional optimal quantum strategy.
When this strategy satisfies the usual implementation requirement, combining it with the quantum completeness of the KLVY compiler~\cite{kalai2023quantum} gives a matching lower bound up to negligible error.
In other words, flat optimality gives quantitative control of the compiled score and not just a soundness upper bound.

Note also that an infinite-dimensional quantum strategy poses several issues. First, it is unclear how to implement such a strategy efficiently with polynomial-size circuits. Second, even if one could engineer such an implementation, compiling it while preserving its score would require a justification of the correctness of the QHE scheme in the infinite-dimensional setting.

\vspace{1.0em}
\noindent
\textit{The sequential Navascu\'es-Pironio-Ac\'in hierarchy.}
As a second main result, we formally introduce and characterize the sequential NPA hierarchy (\cref{sec:SequentialNPA}), which underpins our quantitative soundness proof.
Our analysis establishes its two most important theoretical properties: a stopping criterion based on the flatness condition (aka. rank-loop, a condition on the solution's matrix rank indicating a finite-dimensional solution) and its asymptotic convergence to the commuting score.
We summarize these findings as follows:
\begin{maintheorem}[\cref{thm:StoppingCriteriaSeqNPA,thm:SequentialNPAConvergence}]\label{thm:StoppingCrteriaInformal}
    The sequential NPA hierarchy is guaranteed to converge to the commuting score as $n \to \infty$.
    Moreover, if the hierarchy for a bipartite Bell game $\cG$ admits a flat optimal solution at some finite level $n$, then the flat solution gives rise to a finite-dimensional optimal tensor-product quantum strategy and $\gamevalueSeqNPA{\cG}{n}=\gamevalueqcopt{\cG}=\omega_{\mathrm q}(\cG)$.
\end{maintheorem}
\noindent
In addition, we:
\begin{enumerate}
    \item Establish the precise relationship between the sequential and standard NPA hierarchies, which forms the basis of our convergence proof.
    In \cref{prop:CompareStandardSequentialNPA}, we prove that the sequential NPA hierarchy is equivalent to a modified relaxation of the standard hierarchy (\cref{eq:ModifiedNPASDP}), where Alice's completeness condition is imposed only against Bob's polynomials.
    As formalized in \cref{thm:SequentialNPAConvergence}, these comparisons imply that both the sequential and modified hierarchies converge to the quantum commuting score.
    \item Identify (via \cref{prop:SequentialDualSparse}) its conic dual with the sparse sum of squares (SOS) hierarchy (\cref{eq:sparseSOSSDP})~\cite{klep2022sparse}.
    This duality connects our hierarchy to the theory and computational methods of sparse SOS optimization~\cite{klep2022sparse,magron2023sparse}.
\end{enumerate}

\subsection{Methods, techniques and further results}
Our results rely on a combination of existing tools adapted to the compiled game setting and novel techniques developed in this work, which may be of independent interest.
Key elements include:
\begin{enumerate}
    \item \emph{Navascu\'es-Pironio-Ac\'in hierarchy and its generalizations.} 
    The standard NPA hierarchy~\cite{navascues2008convergent, pironio2010convergent} provides a systematic method, based on semidefinite programming (SDP), to compute upper bounds on the  commuting quantum score $\gamevalueqcopt{\cG}$.
    It involves a sequence of SDP relaxations indexed by an integer level $n$, yielding monotonically decreasing upper bounds $\gamevalueNPA{\cG}{n}$ that converge to $\gamevalueqcopt{\cG}$.
    It generalizes the Lasserre-Parrilo hierarchy~\cite{lasserre2001global, parrilo2003semidefinite} to non-commutative settings.

    As discussed in the previous subsection, to establish our quantitative bound on $\gamevalueCompile{\lambda}{\cG_{\comp}, S}$ (\cref{thm:BoundScoreCompiledInformal}), we consider a sequential generalization of the standard NPA hierarchy, which we introduce and fully characterize (\cref{sec:SequentialNPA}).

    \item \emph{Imperfect finite-dimensional quantum representations via flat extension.}
    To connect finite levels of the (sequential) NPA hierarchy to concrete quantum representations, we consider the \emph{flat extension} method~\cite{helton2012convex}, central to the discussion in ~\cref{sec:FlatExtension,sec:FiniteSecurityGNSStrategy}, and pivotal in the proof of \cref{thm:StoppingCrteriaInformal}, \cref{prop:CompareStandardSequentialNPA}.
    Given the moment matrix from a finite level $n$ solution of the NPA hierarchy, the flat extension technique gives positive linear functionals and, via the GNS construction, a representation of the associated finite-dimensional quantum strategy that exactly satisfies all algebraic constraints imposed by that $n$-th NPA level.

    Notably, while these extracted strategies faithfully realize the $n$-th level NPA model, the constraints of this finite level are generally weaker than those of an ideal commuting quantum strategy.
    For instance, the $n$-th level NPA hierarchy enforces that certain polynomial expressions involving commutators evaluate to zero, as they would for truly commuting operators.
    I.e., $\Tr{\rho [\A, \B] P} = 0$ for all polynomials $P$ of degrees $\leq 2n-2$.
    However, it does not, in general, enforce the operator identity $[\A, \B]=0$.

    Consequently, the strategies obtained via flat extension from a finite NPA level are ``imperfect'' in the sense that Alice's and Bob's operators might not strictly commute with each other, even though all $n$-th level NPA conditions (including those partial commutativity constraints and linear constraints like POVMs summing to identity) are met.
    This technique thus provides a concrete way to construct operational (albeit imperfect) quantum representations from a finite level of the NPA hierarchy.

    It is worth noting that the authors of~\cite{coudron2015interactive} presented an alternative construction of almost commuting strategies from the NPA hierarchy.
    While our flat extension-based method produces strategies satisfying exact commutation when tested against low-degree polynomials their approach yields strategies whose commutators are controlled in operator norm, with a bound scaling as $O(1/\sqrt{n})$ for the $n$-th NPA level.
    This is achieved by analyzing the projections onto low-degree subspaces of the original NPA solution, rather than by constructing a new representation from a modified moment matrix.
    
    \item \emph{Isolating signaling effect using symmetric group representation theory.}
    A key observation from~\cite{kulpe2024bound} is that every QPT strategy of compiled Bell games at security parameter $\lambda$ implicitly contains a negligible amount of signaling (permitted by the QHE scheme) from the protocol's encrypted part to the unencrypted part with $\poly(\lambda)$-size circuits.

    Therefore, analyzing this weak signaling effect and its impact on the compiled Bell score is interesting.
    To this end, inspired by~\cite{renou2017inequivalence}, we utilize representation theory of the symmetric group to develop a technique for decomposing the operators that do not satisfy the ideal no-signaling conditions (\cref{prop:DecompositionNSSI}).
    This method allows us to systematically decompose these operators into a no-signaling part, a signaling part, and a residual term.
    This decomposition is central to proving our main theorem (\cref{thm:BoundScoreCompiledInformal}), since the no-signaling part defines a feasible solution of the sequential NPA hierarchy at a fixed level.
    The signaling contribution to the Bell score is negligible, as is the contribution of the residual term after multiplication by its small coefficient.
    Observe that, since this decomposition technique is formulated rather generally, it may also be useful for isolating and analyzing signaling effects in other quantum protocols.
\end{enumerate}

\subsection{Open problems and outlook}\label{sec:OpenProblemOutlook}
Building on our results, several important open questions for future research emerge:
\begin{enumerate}
    \item \textit{Necessity of the NPA approximation error:}
    A key question arising from our work is whether the game-specific NPA approximation error $\epsilon(n)$ is necessary for quantitative quantum soundness.
    For a bipartite Bell game $\cG$ whose sequential NPA hierarchy does not converge at any finite level, does every QPT strategy family $S$ still admit a negligible function $\negl_S$ such that
    \begin{align}\label{eq:OpenSoundnessWithoutNPAError}
        \gamevalueCompile{\lambda}{\cG_{\comp},S}\leq\gamevalueqcopt{\cG}+\negl_S(\lambda)
    \end{align}
    for all $\lambda$?
    Our results establish such a bound under finite-level convergence, but do not settle this question in general.

    Note that~\cite[Theorem~1.3]{fanizza2025npahierarchydoesattain} proves that there exist Bell games for which the standard NPA hierarchy does not attain the commuting quantum value at any finite level.
    For such a game $\cG$, the comparison in~\cref{sec:SequentialVsStandardNPA} gives
    \begin{align}\label{eq:NoFiniteSequentialConvergence}
        \gamevalueSeqNPA{\cG}{n}\geq\gamevalueNPA{\cG}{n+1}>\gamevalueqcopt{\cG}
    \end{align}
    for every $n\geq1$.
    Thus the sequential NPA hierarchy also fails to attain the commuting quantum value at any finite level. Even though $\epsilon(n)\to0$ as $n\to\infty$, $\epsilon(n)>0$ for every finite $n$.

    \item \textit{Separation between sequential and standard NPA hierarchies:}
    We introduced the sequential NPA hierarchy and showed it is equivalent to the standard NPA hierarchy at finite levels with relaxed POVM completeness constraints.
    We also characterized its stopping criteria and identified conic dual with the sparse SOS hierarchy~\cite{klep2022sparse, magron2023sparse}.
    An interesting question is whether there exist Bell games $\cG$ for which the sequential NPA hierarchy $\gamevalueSeqNPA{\cG}{n}$ converges much slower to $\gamevalueqcopt{\cG}$ than the standard NPA hierarchy $\gamevalueNPA{\cG}{n}$.
    Finding such explicit separations (which we conjecture exist considering the numerical analysis on $I_{3322}$ of the sparse SOS hierarchy~\cite[Chapter~6.7]{magron2023sparse}) would provide deeper insights into the convergence properties of these hierarchies and the precise implications of using Arveson's Radon-Nikodym derivatives~\cite[Lemma~1.4.1]{arveson1969subalgebras}, see \cref{prop:ArvesonRadonNikodym}, in the sequential formulation.

    \item \textit{Generalization to robust self-testing for compiled games:}
    Robust self-testing allows characterizing a quantum device based solely on observed correlations, even with experimental imperfections.
    While the exact self-testing result of compiled Bell games in the asymptotic limit of the security parameter is established~\cite[Theorem~6.5]{kulpe2024bound}, the question of whether one can generalize this to the robust case in the non-asymptotic setup remains open.
    We explore this direction in \cref{sec:RobustSelfTesting}, and note on the need to extend the notion of robust self-testing beyond quantum strategies to cover ``quasi-quantum'' or imperfectly realized strategies, possibly using results similar to~\cite{xu2025quantitative}.

    \item \textit{Quantum soundness of multipartite compiled Bell games beyond two parties:}
    Current investigations into quantum soundness, including our own, have primarily focused on the compilation of bipartite Bell games.
    Extending quantitative quantum soundness results to games with three or more provers is the natural next step, but it presents a significant challenge: it requires a sophisticated generalization of operator-algebraic tools, namely for Arveson's Radon-Nikodym derivatives (\cref{prop:ArvesonRadonNikodym})~\cite[Lemma~1.4.1]{arveson1969subalgebras}.
    
    In concurrent work, the authors of~\cite{baroni2025asymptotic} address this very issue, establishing asymptotic quantum soundness for all multipartite games by proving a new chain rule for these derivatives.
    Their multipartite framework is complementary to our methods, and we believe merging their techniques with ours provides a clear path toward a quantitative quantum soundness analysis for multipartite compiled games.
    
    \item \textit{Bigger picture---from space-like separated provers to single compiled provers:}
    A compelling direction in quantum information involves replacing the requirement of space-like separation in Bell game-based protocols with computational or cryptographic assumptions on a single quantum device.
    Beyond the research on compiled Bell games already discussed, recent works have also advanced our understanding of nonlocality under computational assumptions~\cite{gluch2024nonlocality}, as well as applications in self-testing~\cite{metger2021self} and device-independent quantum key distribution~\cite{metger2021device} in the single-prover paradigm.

    Our work contributes to this broader effort by providing quantitative soundness bounds for all bipartite compiled Bell games.
    More fundamentally, the operator algebraic techniques we employ offer a direct bridge between the ``space-like separation world'' and the ``compiled single-prover world,'' suggesting the potential for a unified mathematical framework.
    Such a framework could systematically translate protocols originally designed for spatially separated parties into equivalent single-prover protocols with cryptographic assumptions, all while quantitatively preserving their essential properties (such as achievable scores).
\end{enumerate}

\subsection{Structure of the paper}


In \cref{sec:QuantitaiveBoundConvergentRateMain}, we establish quantitative upper bounds for the quantum scores of compiled Bell games.
More specifically, \cref{sec:CompileGameSetup} introduces compiled Bell games in the context of the sequential NPA hierarchy at level $n$ and the associated relaxed no-signaling conditions.
\cref{sec:FlatExtension} details the flat extension technique, crucial for extending positive linear maps defined on subspaces of operators to positive linear functionals on the full algebra.
Building on this, \cref{sec:FiniteSecurityGNSStrategy} constructs a quantum representation for these compiled Bell games from the extended functionals.
A key technical contribution is presented in \cref{sec:TechnicalDecomposition}, where we develop a method to decompose Alice's operators into signaling and no-signaling components, allowing us to bound the signaling advantage.
\cref{sec:MainTheorems} then combines these elements to present the main quantitative soundness theorems, relating the compiled game scores to the sequential NPA hierarchy and the quantum scores.
Finally, \cref{sec:RobustSelfTesting} briefly discusses potential notions and challenges for robust self-testing in the context of compiled Bell games.

In \cref{sec:SequentialNPA}, we formally introduce and analyze the sequential NPA hierarchy (\cref{eq:SequentialNPASDP}).
In particular, \cref{sec:SequentialVsStandardNPA} compares this hierarchy to the standard NPA hierarchy, particularly at finite levels where the sequential version is equivalent to the standard NPA hierarchy with a relaxed POVM completeness condition.
The convergence of the sequential NPA hierarchy is then an easy consequence.
In \cref{sec:StoppingCriteriaSeqNPA} we fully describe and prove the stopping criteria of the sequential NPA hierarchy, and provide a concrete example of flat optimal solution with CHSH game.
Finally, \cref{sec:SparseDualSequential} further characterizes the sequential NPA hierarchy by identifying its conic dual as a special case of the sparse sum of squares (SOS) hierarchy.

\section{Quantitative bounds for the compiled scores with the sequential NPA hierarchy}\label{sec:QuantitaiveBoundConvergentRateMain}
This section investigates the relationship between the optimal score of a Bell game $\cG$ in the standard commuting quantum model, denoted $\gamevalueqcopt{\cG}$, and the score achievable by a prover in its compiled version $\cG_{\comp}$ when employing a specific quantum polynomial time (QPT) strategy $S$.
Such a QPT strategy, $S = (S_\lambda)_{\lambda \in \mathds{N}}$, is understood as a sequence of quantum strategies indexed by the security parameter $\lambda$; each $S_\lambda$ consists of quantum operations whose complexity (e.g., in the quantum circuit model) is polynomial in $\lambda$.
(For a detailed definition of such strategies, we refer to~\cite[Definition~4.3]{kulpe2024bound}.)
We denote by $\gamevalueCompile{\lambda}{\cG_{\comp}, S}$ the score achieved by the prover when using the QPT strategy $S = (S_\lambda)_{\lambda \in \mathds{N}}$ in the compiled game $\cG_{\comp}$.

Our analysis is rooted in the \emph{sequential NPA hierarchy} (defined in \cref{eq:SequentialNPASDP}) for the specific game $\cG$.
Let us quantify the gap between the $n$-th level of the sequential NPA hierarchy and the optimal commuting quantum score by defining
\begin{align}\label{eq:epsilon_convergence_assumption}
    \epsilon(n) := \gamevalueSeqNPA{\cG}{n} - \gamevalueqcopt{\cG} \geq 0,
\end{align}
such that $\epsilon(n) \to 0$ as $n \to \infty$ due to the asymptotic convergence of the sequential NPA hierarchy.

Our main findings in this section establish two key quantitative bounds.
First, for every fixed QPT strategy family $S$ and every fixed integer $n\geq1$, we show in \cref{thm:GeneralQuantBoundFiniteSecurityQCValue} that
\begin{align}\label{eq:main_bound_compile_vs_seqNPA}
    \gamevalueCompile{\lambda}{\cG_{\comp}, S} \leq \gamevalueSeqNPA{\cG}{n} + \etaThm_{S, n}(\lambda) = \gamevalueqcopt{\cG} + \epsilon(n) + \etaThm_{S, n}(\lambda).
\end{align}
Here, $\etaThm_{S, n}: \mathds{N} \to \mathds{R}_{\geq 0}$ is a negligible function dependent on the QHE scheme used in the compilation of $\cG_{\comp}$, the QPT strategy $S$, and the sequential NPA hierarchy level $n$.
(Recall that \emph{negligible} means $\etaThm_{S, n}$ goes to zero faster than the reciprocal of any polynomial in $\lambda$.)
This first bound controls the prover's possible advantage over the optimal sequential NPA score at the fixed level $n$.
For each fixed $n$, letting $\lambda \to \infty$ and then $n \to \infty$, we recover the qualitative quantum soundness established in~\cite{kulpe2024bound}.

If the sequential NPA hierarchy converges at a finite level $n_0$, then $\epsilon(n_0)=0$, and \cref{cor:MainBoundFiniteSecurityQCValue} gives $\gamevalueCompile{\lambda}{\cG_{\comp},S}\leq\gamevalueqcopt{\cG}+\etaThm_S(\lambda)$, with $\etaThm_S:=\etaThm_{S,n_0}$.
In particular, if the hierarchy admits a flat optimal solution at level $n_0$, then \cref{thm:StoppingCriteriaSeqNPA} gives $\gamevalueqcopt{\cG}=\omega_{\mathrm{q}}(\cG)$, and hence
\begin{align}\label{eq:main_bound_intro_sec2}
    \gamevalueCompile{\lambda}{\cG_{\comp}, S} \leq \omega_{\mathrm{q}}(\cG) + \etaThm_{S}(\lambda),
\end{align}
where $\omega_{\mathrm{q}}(\cG)$ is the optimal finite-dimensional quantum score and $\etaThm_{S}(\lambda)$ is a negligible function depending on the QHE scheme and the QPT strategy $S$.
This is a generalization of~\cite{natarajan2023bounding,cui2024computational,baroni2024quantum,mehta2024self}.

To facilitate the analysis, we fix the hierarchy level $n$ independently of $\lambda$ in \cref{sec:CompileGameSetup} and make the polynomial dependence of the negligible estimates explicit.

The section is organized as follows.
\cref{sec:CompileGameSetup} reviews the relevant definitions for compiled Bell games and the polynomial estimates at level $n$ of the sequential NPA hierarchy.
\cref{sec:FlatExtension} constructs positive extensions of the functionals to the full algebra $\cB$.
In \cref{sec:FiniteSecurityGNSStrategy}, we perturb these functionals negligibly by a common faithful state and apply the averaging and GNS construction to obtain a quantum representation whose correlation $\widetilde p^{\lambda,n}$ differs negligibly from $p^\lambda$.
In \cref{sec:TechnicalDecomposition}, we decompose Alice's operators into signaling, no-signaling, and residual components and establish positivity of the no-signaling part.
This result enables us to bound the potential signaling effect from the encrypted part of the prover to the unencrypted part, while associating the no-signaling part with the strongly no-signaling sequential NPA hierarchy at level $n$.
The technique of bounding weak signaling effects might be interesting beyond the scope of compiled Bell games.
Then, \cref{sec:MainTheorems} combines the perturbation, signaling, and residual estimates to prove the main bound (\cref{thm:GeneralQuantBoundFiniteSecurityQCValue}) and the finite-convergence corollary (\cref{cor:MainBoundFiniteSecurityQCValue}).
We finish in \cref{sec:RobustSelfTesting} with a discussion on robust self-testing of compiled Bell games and the relation between exact commuting operator self-testing and the robust condition in \cref{eq:RobustCoefficientDegreeBound}.

\subsection{Compiled Bell games and QPT strategies associated with NPA level \texorpdfstring{$n$}{n}}\label{sec:CompileGameSetup}

We begin with a compiled Bell game $\cG_{\comp}$ where the verifier selects the security parameter $\lambda$, and considers an arbitrary QPT strategy $S = (S_{\lambda})_{\lambda}$ with correlations $(p^{\lambda}(ab|xy))_{\lambda}$ for input-output $(a, b, x, y) \in I_{A} \times I_B \times I_X \times I_Y$.
If $p^{\lambda}(a|x) = \sigma_{a|x}^{\lambda}(1) = 0$, Cauchy--Schwarz inequality implies that $\sigma_{a|x}^{\lambda} = 0$.
Thus we may, without loss of generality, assume that $p^{\lambda}(a|x) \neq 0$

By the results in~\cite{kulpe2024bound}, we can interpret the game and QPT strategy as a sequential Bell game $\cG_\seq$ with a relaxed no-signaling condition (\cref{eq:securityassumptionNPA}).
In their notation, let $\cB$ be the unital $C^*$-algebra generated by Bob's POVM elements $\{\B\}$ (for output-input pairs $(b,y) \in I_B \times I_Y)$, subjetct to $\B \geq 0$ and $\sum_b \B = \id$.
Then for the output-input pairs $(a,x) \in I_A \times I_X$, the measurements of the strategy $S_{\lambda}$ are captured by the positive linear functionals
\begin{align}
    \sigma^{\lambda}_{a|x}: \cB \to \mathds{C}, \forall a, x, \text{ s.t.} \ p^{\lambda}(ab|xy) = \sigma^{\lambda}_{a|x}(\B).
\end{align}
Moreover, the marginalization over $a$ gives the states (i.e., normalized positive linear functionals) $\sigma^{\lambda}_x: \cB \to \mathds{C}$ for all $x$ via
\begin{align}
    \sigma^{\lambda}_{x} := \sum_a \sigma^{\lambda}_{a|x}.
\end{align}
Then, by~\cite[Proposition~4.6]{kulpe2024bound}, for every fixed polynomial $P$, there exists a negligible function $\eta_P(\lambda)$ such that 
\begin{align}\label{eq:securityassumptionOriginal}
    \lvert (\sigma^{\lambda}_{x} - \sigma^{\lambda}_{x'})(P) \rvert \leq \eta_P(\lambda),
\end{align}
where $\eta_P$ depends on the specific polynomial $P$, the QHE scheme, and the QPT strategy $S$.
Note that this inequality does not imply there is a universal $\eta$ providing a uniform bound for all $P$.
In the asymptotic limit of security parameter $\lambda \to \infty$ (hence $\eta(\lambda) \to 0$), one recovers the \emph{strongly no-signaling sequential algebraic strategy}~\cite[Definition~5.14]{kulpe2024bound}.

The physical intuition remains relevant: a prover implementing $S_\lambda$ is, by definition, restricted to computations (and thus, state preparations and measurements) whose complexity is bounded by $\poly(\lambda)$.
It is therefore natural to analyze $S_\lambda$ not against arbitrarily complex quantum measurements, but rather by considering its interaction with observables whose complexity is also bounded.
This motivates our choice to focus our analysis on a specific set of polynomials, namely those relevant to a particular level of the NPA hierarchy.

Fix an integer $n \geq 1$ independently of $\lambda$, and keep the QPT strategy family S fixed throughout the argument. 
Instead of the full $C^*$-algebra $\cB$, we restrict our attention to the $2n$-degree subspace $\cB_{2n} = \{ P(\{\B\}) \mid \deg(P) \leq 2n \}$.
This perspective aligns naturally with the sequential NPA hierarchy (formally defined in \cref{eq:SequentialNPASDP}), where our $n$ corresponds to the $n$-th level of this hierarchy.
The identification with the sequential NPA hierarchy at finite level is precisely what ensures the validity of our signaling decomposition technique (\cref{prop:DecompositionNSSI}).

In this level $n$ sequential NPA context, we naturally consider the restriction of $\sigma^{\lambda}_{a|x}$ to $\cB_{2n}$.
That is, for the output-input pairs $(a, x)$ for Alice, the measurements of the strategy $S_{\lambda}$ are captured by the positive linear maps (rather than functionals on the full $\cB$)
\begin{align}
    \sigma^{\lambda, n}_{a|x}: \cB_{2n} \to \mathds{C}, \forall a, x, \text{ s.t.}\ p^{\lambda}(ab|xy) = \sigma^{\lambda, n}_{a|x}(\B).
\end{align}
Similarly, marginalization over $a$ gives normalized linear maps (rather than states) $\sigma^{\lambda, n}_{x}: \cB_{2n} \to \mathds{C}$ for all $x$, in the sense that
\begin{align}
    \sigma^{\lambda, n}_{x} := \sum_a \sigma^{\lambda, n}_{a|x}.
\end{align}
It directly follows from \cref{eq:securityassumptionOriginal}, for all $P \in \cB_{2n}$, we have weakly no-signaling constraints as
\begin{align}\label{eq:securityassumptionNPA}
    \lvert (\sigma^{\lambda, n}_{x} - \sigma^{\lambda, n}_{x'})(P) \rvert \leq \eta_P(\lambda).
\end{align}

We further refine the no-signaling constraints.
Consider the finite set of monomials in the generators $\{B_{b|y}\}$ of degree at most $2n$, including $\id$.
These monomials belong to $\cB_{2n}$ and span it.
Define
\begin{align}\label{eq:ReviseDegnNegFunc}
    \eta_n(\lambda) := \max_{x,x'\in I_X}\max_{\substack{w\in\cB_{2n}\\w\text{ a monomial}}}|(\sigma^{\lambda,n}_x-\sigma^{\lambda,n}_{x'})(w)|.
\end{align}
By \cref{eq:securityassumptionNPA}, this maximum is negligible because the monomials and question pairs form a fixed finite set.
For every polynomial $P \in \cB_{2n}$ such that $P=\sum_w c_w w$ in these monomials of $\{\B\}$, write $\coeffnorm{P} := \sum_w|c_w|$.
Linearity and the triangle inequality give
\begin{align}\label{eq:ReviseDegnNegPolyCoeff}
    |(\sigma^{\lambda,n}_x-\sigma^{\lambda,n}_{x'})(P)|\leq\sum_w|c_w|\,|(\sigma^{\lambda,n}_x-\sigma^{\lambda,n}_{x'})(w)|\leq\eta_n(\lambda)\coeffnorm{P}.
\end{align}
The right-hand side is negligible for each fixed $\lambda$-independent $P$.

\subsection{Flat extension to functionals on full algebra}\label{sec:FlatExtension}
Analogously to~\cite{kulpe2024bound}, we wish to apply Arveson's Radon-Nikodym Theorem~\cite[Lemma~1.4.1]{arveson1969subalgebras}, to obtain a commuting quantum strategy corresponding to $p^{\lambda}(ab|xy)$.
Let us first recall this key mathematical result.
\begin{proposition}[Arveson's Radon-Nikodym derivative]\label{prop:ArvesonRadonNikodym}
    Let $\omega, \nu$ be positive linear functionals on a unital $C^*$-algebra $\cB$ such that $\nu \leq \omega$, and let $(\cH_{\omega}, \pi_{\omega}, \ket{\Omega_{\omega}})$ be the GNS triple of $\omega$.
    Then there exists a unique operator $T \in \pi_{\omega}(\cB)'$ such that $0 \leq T \leq \id_{\cH_{\omega}}$ and
    \begin{align}
        \nu(P) = \sandwich{\Omega_{\omega}}{T \pi_{\omega}(P)}{\Omega_{\omega}}
    \end{align}
    for all $P \in \cB$.
\end{proposition}
\begin{proof}
    This is a special case of the general completely positive map version of~\cite[Theorem~1.4.2]{arveson1969subalgebras}.
\end{proof}

However, one difficulty of directly applying \cref{prop:ArvesonRadonNikodym} is that the maps $\sigma^{\lambda, n}_{a|x}$ from \cref{sec:CompileGameSetup} are, for each $(a,x)$, positive linear maps on the subspace $\cB_{2n}$ of polynomials in $\{\B\}$ of degree up to $2n$, rather than functionals on the full $C^*$-algebra $\cB$.
We address this by extending each $\sigma^{\lambda, n}_{a|x}$ to a positive linear functional on $\cB$ using a \emph{flat extension technique}, similar to that in~\cite[Proposition~2.5 \& Remark~2.6]{helton2012convex}.
The method is rooted in the following characterization of positive semidefinite (PSD) block matrices.
\begin{proposition}\label{prop:FlatExtensionLinearAlgebraFact}
    Let
    \begin{align}
        \Tilde{A} = \begin{pmatrix}
            A & B \\
            B^* & C
        \end{pmatrix}
    \end{align}
    be a self-adjoint matrix.
    Then $\Tilde{A} \succeq 0$ if and only if $A \succeq 0$, and there exists some matrix $Z$ with $B = A Z$ and $C \succeq Z^* A Z$.
    A crucial consequence is that the specific choice $C = Z^* A Z$ makes the matrix
    \begin{align}
        M_{f} = \begin{pmatrix}
            A & B \\
            B^* & Z^* A Z
        \end{pmatrix}
    \end{align}
    PSD, and importantly, $\mathrm{rank}(M_{f}) = \mathrm{rank}(A)$, i.e., $M_{f}$ is flat over $A$.
    The matrix $Z$ can be generally computed using the (e.g., Moore-Penrose) inverse of $A$ due to $\mathrm{range}(B) \subset \mathrm{range}(A)$.
\end{proposition}
\begin{proof}
    See~\cite[Proposition~1.11]{burgdorf2016optimization} (adapted to complex matrices).
\end{proof}

For the construction that follows, we may assume that our initial positive linear maps $\sigma^{\lambda, n}_{a|x}$ are defined on the slightly larger subspace $\cB_{2n+2}$, ensuring cleaner notation.
For each $(a,x)$, we associate the map $\sigma^{\lambda, n}_{a|x}: \cB_{2n+2} \to \mathds{C}$ with its corresponding moment (or Hankel) matrix, indexed by the monomials in the generators $\{\B\}$.
In particular, for $k \leq n+1$, denote by $M_k(\sigma^{\lambda, n}_{a|x})$ the $k$-th order moment matrix defined by
\begin{align}\label{eq:MomentMatrixStateIdentification}
    (M_k(\sigma^{\lambda, n}_{a|x}))_{w, v} = \sigma^{\lambda, n}_{a|x}(w^* v)
\end{align}
for monomials $w, v \in \cB_{k}$.
It is straightforward to check $\sigma^{\lambda, n}_{a|x}$ is positive if and only if $M_k(\sigma^{\lambda, n}_{a|x}) \succeq 0$ for every $k \leq n+1$.

The $(n+1)$-th order moment matrix, $M_{n+1}(\sigma^{\lambda, n}_{a|x})$, can then be written in block form:
\begin{align}
    M_{n+1}(\sigma^{\lambda, n}_{a|x}) = \begin{pmatrix}
        M_n(\sigma^{\lambda, n}_{a|x}) & B \\
        B^* & C
    \end{pmatrix},
\end{align}
where the block $B$ has entries $\sigma^{\lambda, n}_{a|x}(w^*v)$ for monomials $w \in \cB_n$ and $v \in \cB_{n+1}\setminus\cB_n$, while $C$ has entries defined by monomials of degree exactly $n+1$.
\cref{prop:FlatExtensionLinearAlgebraFact} then implies that we can construct a matrix $Z$ such that $B = M_n(\sigma^{\lambda, n}_{a|x}) Z$ and a new PSD $(n+1)$-th order moment matrix
\begin{align}
    M_{n+1}(\tilde{\sigma}^{\lambda, n}_{a|x}) = \begin{pmatrix}
        M_n(\sigma^{\lambda, n}_{a|x}) & B \\
        B^* & Z^* M_n(\sigma^{\lambda, n}_{a|x}) Z
    \end{pmatrix} \succeq 0.
\end{align}

This moment matrix $M_{n+1}(\tilde{\sigma}^{\lambda, n}_{a|x})$, as suggested by its notation, can be identified with a new positive linear map $\tilde{\sigma}^{\lambda, n}_{a|x}: \cB_{2n+2} \to \mathds{C}$ via \cref{eq:MomentMatrixStateIdentification}.
This new map $\tilde{\sigma}^{\lambda, n}_{a|x}$ agrees with the original ${\sigma}^{\lambda, n}_{a|x}$ on $\cB_{2n+1}$ (since the blocks $M_n(\sigma^{\lambda, n}_{a|x})$ and $B$ are preserved) but generally differs on $\cB_{2n+2}\setminus\cB_{2n+1}$ due to the modified bottom-right block.
Moreover, $M_{n+1}(\tilde{\sigma}^{\lambda, n}_{a|x})$ by construction satisfies the \emph{flatness condition} (also called rank-loop condition, cf.~\cite{navascues2008convergent}),
\begin{align}\label{eq:FlatConditionSec2}
    \mathrm{rank}(M_{n+1}(\tilde{\sigma}^{\lambda, n}_{a|x})) = \mathrm{rank}(M_n({\sigma}^{\lambda, n}_{a|x})),
\end{align}
which is the key to constructing a finite-dimensional representation as the following.

\begin{proposition}\label{prop:GNSfromFlatExtension}
    Given the positive linear map $\tilde{\sigma}^{\lambda, n}_{a|x}: \cB_{2n+2} \to \mathds{C}$ with its $(n+1)$-th order flat moment matrix $M_{n+1}(\tilde{\sigma}^{\lambda, n}_{a|x})$ constructed as above, and letting $p^{\lambda}(a|x) = \sigma^{\lambda, n}_{a|x}(\id) \neq 0$.
    Then, there exists a finite-dimensional GNS representation $(\cH^{\lambda, n}_{a|x}, \pi^{\lambda, n}_{a|x}, \ket{\Omega^{\lambda, n}_{a|x}})$ of the C*-algebra $\cB$ such that:
    \begin{enumerate}[label=(\roman*)]
        \item The Hilbert space $\cH^{\lambda, n}_{a|x}$ has dimension $\mathrm{rank}(M_n(\sigma^{\lambda, n}_{a|x}))$. It is spanned by vectors corresponding to polynomials up to degree $n$:
        \begin{align}
            \cH^{\lambda, n}_{a|x} = \mathrm{span}\{\pi^{\lambda, n}_{a|x}(P) \ket{\Omega^{\lambda, n}_{a|x}} \mid P \in \cB_{n} \}.
        \end{align}
        Consequently, for any polynomial $P \in \cB$, there exists $P' \in \cB_n$ such that $\pi^{\lambda, n}_{a|x}(P) \ket{\Omega^{\lambda, n}_{a|x}} = \pi^{\lambda, n}_{a|x}(P') \ket{\Omega^{\lambda, n}_{a|x}}$.
        \item The map $\tilde{\sigma}^{\lambda, n}_{a|x}$ (and thus $\sigma^{\lambda, n}_{a|x}$ on $\cB_{2n+1}$) is recovered by the cyclic vector:
        for all $P \in \cB_{2n+1}$,
        \begin{align}
            \sigma^{\lambda, n}_{a|x}(P) = \sigma^{\lambda, n}_{a|x}(\id) \cdot \sandwich{\Omega^{\lambda, n}_{a|x}}{\pi^{\lambda, n}_{a|x}(P)}{\Omega^{\lambda, n}_{a|x}}.
        \end{align}
        \item The representation preserves the POVM structure of the generators: for each $y$, the set $\{\pi^{\lambda, n}_{a|x}(\B)\}_{b}$ forms a POVM on $\cH^{\lambda, n}_{a|x}$
        (higher order constraints, such as commutativity, are not necessarily preserved, but this is not required for our current purpose).
    \end{enumerate}
\end{proposition}
\begin{proof}
The representation is obtained by applying the standard GNS construction to the normalized map $\tilde{\sigma}^{\lambda, n}_{a|x} / p^{\lambda}(a|x)$. The main consideration, differing from the GNS construction for a state on the full algebra $\cB$, is that $\tilde{\sigma}^{\lambda, n}_{a|x}$ is initially defined only on $\cB_{2n+2}$.
This limitation requires extra care to ensure that the representation operators $\pi^{\lambda, n}_{a|x}(X)$ (defined by left multiplication) are well-defined, i.e., that they map the GNS Hilbert space $\mathcal{H}^{\lambda, n}_{a|x}$ to itself.
Thankfully, the flatness condition on the moment matrix $M_{n+1}(\tilde{\sigma}^{\lambda, n}_{a|x})$ guarantees this well-definedness, effectively through rank and dimension constraints, allowing $\pi^{\lambda, n}_{a|x}$ to be a *-representation of the whole $\cB$.
The properties (i)-(iii) then follow.
For detailed arguments, see e.g.,~\cite[Proposition~2.5 \& Remark~2.6]{helton2012convex} or~\cite[Theorem~10]{navascues2008convergent}.
\end{proof}

Thus \cref{prop:GNSfromFlatExtension} allows us to consistently extend $\tilde{\sigma}^{\lambda, n}_{a|x}$ (and thereby the original $\sigma^{\lambda, n}_{a|x}$) to a positive linear functional on the entire algebra $\cB$ via the formula:
\begin{equation}
    \begin{aligned}
        \sigma^{\lambda, n}_{a|x}: \cB &\to \mathds{C} \\
        P &\mapsto \sigma^{\lambda, n}_{a|x}(\id) \cdot \sandwich{\Omega^{\lambda, n}_{a|x}}{\pi^{\lambda, n}_{a|x}(P)}{\Omega^{\lambda, n}_{a|x}}.
    \end{aligned}
\end{equation}
Here, and for the rest of this section, we abuse notation by using $\sigma^{\lambda, n}_{a|x}$ to refer to this extended linear functional on $\cB$ as well.

Finally, we define for each of Alice's inputs $x$:
\begin{align}
    \sigma^{\lambda, n}_x = \sum_a \sigma^{\lambda, n}_{a|x}: \cB \to \mathds{C}.
\end{align}
These are indeed states on $\cB$.
Positivity follows from being a sum of positive linear functionals.
Normalization, $\sigma^{\lambda, n}_x(\id)=1$, holds because they are extensions of the original $\sigma^{\lambda, n}_x$ which were normalized on $\cB_{2n}$.
Furthermore, since the extension agrees with the original map on $\cB_{2n+1}$ (and thus on $\cB_{2n}$), the property from \cref{eq:securityassumptionNPA} is preserved: for each $P \in \cB_{2n}$, there exists a negligible function $\eta_P(\lambda)$, dependent on the QHE scheme and $S$, such that
\begin{align}
    \lvert (\sigma^{\lambda, n}_{x} - \sigma^{\lambda, n}_{x'})(P) \rvert \leq \eta_P(\lambda).
\end{align}

The flat extension procedure can be interpreted physically: for each of Alice's outcome-input pairs $(a,x)$, Bob analyzes the correlations $\sigma^{\lambda, n}_{a|x}$ restricted to his measurements corresponding to polynomials up to degree $2n+1$.
He then constructs a minimal (finite-dimensional) quantum model $(\cH^{\lambda, n}_{a|x}, \pi^{\lambda, n}_{a|x}, \ket{\Omega^{\lambda, n}_{a|x}})$ consistent with these observations.
This model then allows extrapolation to define $\sigma^{\lambda, n}_{a|x}$ for any polynomial in Bob's measurements.

We finish the subsection with a remark on the choice of flat extension technique.
\begin{remark}
    To extend $\sigma^{\lambda, n}_{a|x}$ from $\cB_{2n}$ (or $\cB_{2n+2}$) to $\cB$, one might observe that $\cB_k$ forms an operator system for any $k$ and be tempted to apply Arveson's Extension Theorem~\cite[Theorem~7.5]{paulsen2002completely} (or Krein's Theorem for functionals~\cite[Exercise~2.10]{paulsen2002completely}) for this purpose.
    However, these theorems require $\sigma^{\lambda, n}_{a|x}$ to be positive on the $C^*$-algebraic positive cone intersected with the subspace, i.e., on $\cB_+ \cap \cB_{2n+2}$.
    In our setup $\sigma^{\lambda, n}_{a|x}$ is a positive linear map on $\cB_{2n+2}$, meaning that $\sigma^{\lambda, n}_{a|x}$ is positive with respect to sums-of-squares (SOS) $\sigma^{\lambda, n}_{a|x}(\sum_i P_i^* P_i) \geq 0$ for all $P_i \in \cB_{n+1}$.
    The condition, $\sigma^{\lambda, n}_{a|x}(Q) \geq 0$ for all $Q \in \cB_+ \cap \cB_{2n+1}$, is generally stronger, since an element $Q \in \cB_+ \cap \cB_{2n+2}$ might not be a SOS of polynomials in $\cB_{n+1}$ but of much larger degrees.
    Therefore, the positivity condition we start with might be too weak for a direct application of Krein's or Arveson's Extension type theorems, leading us to use the flat extension technique, which guarantees a positive (and state-like after normalization) extension to the whole algebra $\cB$.
\end{remark}

\subsection{Quantum representation for strategies of compiled Bell games}\label{sec:FiniteSecurityGNSStrategy}

Having constructed the states $\sigma^{\lambda, n}_x: \cB \to \mathds{C}$ in \cref{sec:FlatExtension}, we now seek a quantum representation whose correlation differs negligibly from the original compiled correlation $p^{\lambda}$.
We will then use the signaling/non-signaling decomposition in \cref{prop:DecompositionNSSI} to obtain a negligibly close correlation feasible for the sequential NPA hierarchy at level $n$.
To ensure the positivity estimates required by this decomposition, we first perturb the functionals $\sigma^{\lambda, n}_{a|x}$ negligibly by a common state $\rho$.

Since $\cB$ is generated by finitely many elements, it is separable and admits a faithful representation on a separable Hilbert space~\cite[Corollary~II.6.4.10]{blackadar2006operator}, and hence a faithful state.
Fix $\rho:\cB\to\mathds{C}$ to be such a faithful state, independently of $\lambda$ and $S$.
Let $\eta'_n(\lambda)\coloneqq\sqrt{\eta_n(\lambda)}$ for $\eta_n$ from \cref{eq:ReviseDegnNegFunc}, which is negligible for fixed $n$, and define
\begin{align}\label{eq:Reviseregularization}
    \widetilde{\sigma}^{\lambda,n}_{a|x} := \frac{\sigma^{\lambda,n}_{a|x}+\eta'_n(\lambda)\rho/|I_A|}{1+\eta'_n(\lambda)}.
\end{align}
The modified functionals are positive and their marginals $\widetilde{\sigma}^{\lambda,n}_x := \sum_a\widetilde{\sigma}^{\lambda,n}_{a|x}$ are normalized.
The following proposition details the construction of an appropriate representation, generalizing from~\cite[Theorem~5.15]{kulpe2024bound}.

\begin{proposition}\label{prop:GNSConstructionCommonSpace}
    Fix $n\geq1$.
    Let $\{\sigma^{\lambda,n}_{a|x}:\cB\to\mathds{C}\}_{a,x}$ be the states derived from the QPT strategy $S_\lambda$ at NPA level $n$ with the negligible function $\eta_n$ from \cref{eq:ReviseDegnNegFunc}, as constructed in \cref{sec:FlatExtension}.
    Let $\widetilde{\sigma}^{\lambda,n}_{a|x}$ be defined by \cref{eq:Reviseregularization}, with the fixed faithful state $\rho$ and $\eta'_n=\sqrt{\eta_n}$.
    Then there exists a cyclic representation $(\cH^\lambda_n,\pi^\lambda_n,\ket{\Omega^\lambda_n})$ of $\cB$ such that:
    \begin{enumerate}[label=(\roman*)]
        \item There exist positive operators $\{ \hA^{\lambda,n} \}_{a,x}\subset\pi^\lambda_n(\cB)'\subset B(\cH^\lambda_n)$, where $\pi^\lambda_n(\cB)'$ is the commutant of $\pi^\lambda_n(\cB)$.
        
        \item Bob's measurements in this representation, $\{ \pi^\lambda_n(\B) \}_{b,y}$, are POVMs.
        Alice's operators are almost-POVM in the sense that for every $x\in I_X$ and $P_1,P_2\in\cB_n$,
        \begin{align}\label{eq:AAlmostPOVM}
            \left|\sandwich{\Omega^\lambda_n}{\pi^\lambda_n(P_1)\left(\sum_a\hA^{\lambda,n}-\id\right)\pi^\lambda_n(P_2)}{\Omega^\lambda_n}\right| \leq\frac{\eta_n(\lambda)}{1+\eta'_n(\lambda)}\coeffnorm{P_1}\coeffnorm{P_2}.
        \end{align}
        \item Define $V_n:=\operatorname{span}\{ \pi^\lambda_n(P)\ket{\Omega^\lambda_n}\mid P\in\cB_n \}$.
        There is a constant $C_n>0$, independent of $\lambda$ and $S$, such that $\etaLem_n(\lambda):= C_n\eta'_n(\lambda)$ satisfies
        \begin{align}\label{eq:t1-relative-tests}
            |(\widetilde{\sigma}^{\lambda,n}_x-\widetilde{\sigma}^{\lambda,n})(R)| \leq\dim(V_n)\etaLem_n(\lambda)\widetilde{\sigma}^{\lambda,n}(R),
        \end{align}
        where $\widetilde{\sigma}^{\lambda,n}:=|I_X|^{-1}\sum_x\widetilde{\sigma}^{\lambda,n}_x$, for every $x\in I_X$ and every test $R=P^*P$ with $P\in\cB_n$, or $R=Q^*B_{b|y}Q$ or $R=Q^*(\id-B_{b|y})Q$ with $Q\in\cB_{n-1}$ and $(b,y)\in I_B\times I_Y$.
        \item For every $a,b,x,y$, consider the modified correlations
        \begin{align}\label{original-equation-10}
            \widetilde p^{\lambda,n}(ab|xy) :=\widetilde{\sigma}^{\lambda,n}_{a|x}(B_{b|y}) =\sandwich{\Omega^\lambda_n}{\hA^{\lambda,n}\pi^\lambda_n(B_{b|y})}{\Omega^\lambda_n}.
        \end{align}
        For the compiled correlation $p^\lambda(ab|xy)=\sigma^{\lambda,n}_{a|x}(B_{b|y})$, we have
        \begin{equation}\label{eq:t1-correlation-error}
            \begin{aligned}
                \widetilde p^{\lambda,n}(ab|xy)
                &=\frac{p^\lambda(ab|xy)+\eta'_n(\lambda)\rho(B_{b|y})/|I_A|}{1+\eta'_n(\lambda)},\\
                |\widetilde p^{\lambda,n}(ab|xy)-p^\lambda(ab|xy)|
                &\leq\eta'_n(\lambda).
            \end{aligned}
        \end{equation}
    \end{enumerate}
\end{proposition}

\begin{proof}
    In contrast to the strongly no-signaling scenario in~\cite{kulpe2024bound}, where a single state $\sigma$ sufficed to unambiguously form a commuting quantum strategy for $\cG$ via GNS construction, our scenario has many different states $\widetilde{\sigma}^{\lambda,n}_x$.
    As a result, to construct one representation, we must choose a representative state $\widetilde{\sigma}^{\lambda,n}$ that best captures the behavior of all $\widetilde{\sigma}^{\lambda,n}_x$.
    To achieve this, we consider the average state over all $\widetilde{\sigma}^{\lambda,n}_x$,
    \begin{align}\label{eq:AverageOfStateX}
        \widetilde{\sigma}^{\lambda,n} = \frac{1}{|I_X|} \sum_x  \widetilde{\sigma}^{\lambda,n}_{x}.
    \end{align}
    As term with $\rho$ added in \cref{eq:Reviseregularization} cancel when we subtract two marginals, for every polynomial $P\in\cB_{2n}$, one has
    \begin{equation}\label{eq:t1-average-coeff-bound}
        \begin{aligned}
            |(\widetilde{\sigma}^{\lambda,n}-\widetilde{\sigma}^{\lambda,n}_{x'})(P)|
            \leq\frac{1}{|I_X|}\sum_x|(\widetilde{\sigma}^{\lambda,n}_x-\widetilde{\sigma}^{\lambda,n}_{x'})(P)| \leq\frac{\eta_n(\lambda)}{1+\eta'_n(\lambda)}\coeffnorm{P}.
        \end{aligned}
    \end{equation}
    We now construct the GNS triple $(\cH^{\lambda}_{n}, \pi^{\lambda}_{n}, \ket{\Omega^{\lambda}_{n}})$ for this average state $\widetilde{\sigma}^{\lambda,n}$.
    This will be the desired representation.
    Clearly Bob's operators in this representation $\pi^{\lambda}_{n}(\B)$ form POVMs due to the property of $\pi^{\lambda}_{n}$.

    Let us construct Alice's operators $\hA^{\lambda, n}$ acting on $\cH^{\lambda}_{n}$.
    To this end, also consider GNS triples $(\cH^{\lambda}_{x, n}, \pi^{\lambda}_{x, n}, \ket{\Omega^{\lambda}_{x, n}})$ for each $\widetilde{\sigma}^{\lambda,n}_x$.
    Since for each $x$ we have $\widetilde{\sigma}^{\lambda,n}_{a|x} \leq \widetilde{\sigma}^{\lambda,n}_{x}$, \cref{prop:ArvesonRadonNikodym} ensures the existence of POVMs $\{ \hA^{\lambda, x, n} \} \subset \pi^{\lambda}_{x, n}(\cB)' \subset B(\cH^{\lambda}_{x, n})$ such that
    \begin{align}\label{eq:t1-GNS-support-1}
        \widetilde p^{\lambda,n}(ab|xy) = \widetilde{\sigma}^{\lambda,n}_{a|x}(\B) = \sandwich{\Omega^{\lambda}_{x, n}}{\hA^{\lambda, x, n} \pi^{\lambda}_{x, n}(\B)}{\Omega^{\lambda}_{x, n}}.
    \end{align}
    The obstacle is that these POVMs $\{ \hA^{\lambda, x, n} \}$ all act on different Hilbert spaces rather than on $\cH^{\lambda}_{n}$.

    The remedy is to construct, for each $x$, an intertwiner $W^\lambda_{x,n}:\cH^\lambda_n\to\cH^\lambda_{x,n}$, initially defined on the dense cyclic subspace by
    \begin{align}\label{eq:t1-GNS-support-2}
        W^\lambda_{x,n}\pi^\lambda_n(P)\ket{\Omega^\lambda_n}=\pi^\lambda_{x,n}(P)\ket{\Omega^\lambda_{x,n}},
    \end{align}
    for arbitrary $P \in \mathcal{B}$.
    The well-definedness of each $W^{\lambda}_{x, n}$ is ensured because the null ideal of $\widetilde{\sigma}^{\lambda,n}$ (i.e., $\{P \in \cB \mid \widetilde{\sigma}^{\lambda,n}(P^*P) = 0\}$) coincides with the intersection of the null ideals of all $\widetilde{\sigma}^{\lambda,n}_{x}$ (i.e., $\bigcap_x \{P \in \cB \mid \widetilde{\sigma}^{\lambda,n}_{x}(P^*P) = 0\}$) due to \cref{eq:AverageOfStateX} and positivity.
    This guarantees that zero vectors in the GNS representation of $\widetilde{\sigma}^{\lambda,n}$ are mapped to zero vectors in the GNS representations of $\widetilde{\sigma}^{\lambda,n}_{x}$.
    It is also bounded by construction and extends to a bounded operator on the full space.

    Using these intertwiners, we then define Alice's measurement operators as
    \begin{align}\label{original-equation-12}
        \hA^{\lambda, n} &= (W^{\lambda}_{x, n})^*\hA^{\lambda, x, n} W^{\lambda}_{x, n},
    \end{align}
    By construction, one can directly check the correlation identity in statement \textit{(iv)}
    \begin{align}\label{eq:t1-GNS-support-3}
        \widetilde p^{\lambda,n}(ab|xy) = \sandwich{\Omega^{\lambda}_{n}}{\hA^{\lambda, n} \pi^{\lambda}_{n}(\B)}{\Omega^{\lambda}_{n}}.
    \end{align}
    The formula and bound in \cref{eq:t1-correlation-error} follow by evaluating \cref{eq:Reviseregularization} at $B_{b|y}$, since both $p^\lambda(ab|xy)$ and $\rho(B_{b|y})/|I_A|$ lie in $[0,1]$.

    Next, we show the above operators satisfy statement \textit{(i)}, i.e., $\{ \hA^{\lambda, n} \} \subset \pi^{\lambda}_{n}(\cB)'$.
    This relies on the intertwining property of $W^{\lambda}_{x, n}$, namely
    \begin{align}\label{eq:t1-GNS-support-4}
        W^{\lambda}_{x, n} \pi^{\lambda}_{n}(P) = \pi^{\lambda}_{x, n}(P) W^{\lambda}_{x, n}, \, \pi^{\lambda}_{n}(P) (W^{\lambda}_{x, n})^* = (W^{\lambda}_{x, n})^* \pi^{\lambda}_{x, n}(P).
    \end{align}
    The first equality, for example, can be seen from
    \begin{equation}\label{eq:t1-GNS-support-5}
        \begin{aligned}
            W^{\lambda}_{x, n} \pi^{\lambda}_{n}(P_1) \pi^{\lambda}_{n}(P_2) \ket{\Omega^{\lambda}_{n}} &= W^{\lambda}_{x, n} \pi^{\lambda}_{n}(P_1 P_2) \ket{\Omega^{\lambda}_{n}} \\
            &= \pi^{\lambda}_{x, n}(P_1 P_2) \ket{\Omega^{\lambda}_{x, n}} \\
            &= \pi^{\lambda}_{x, n}(P_1) \pi^{\lambda}_{x, n}(P_2) \ket{\Omega^{\lambda}_{x, n}} = \pi^{\lambda}_{x, n}(P_1) W^{\lambda}_{x, n} \pi^{\lambda}_{n}(P_2) \ket{\Omega^{\lambda}_{n}},
        \end{aligned}
    \end{equation}
    for any $P_1, P_2 \in \cB$ of arbitrary degrees, and the cyclicity of $\ket{\Omega^{\lambda}_{n}}$.
    The second equality can be checked similarly.
    Using these intertwining relations and the fact that $\{ \hA^{\lambda, x, n} \} \subset \pi^{\lambda}_{x, n}(\cB)'$, a direct computation shows
    \begin{equation}\label{eq:t1-GNS-support-6}
        \begin{aligned}
            \hA^{\lambda, n}\pi^{\lambda}_{n}(P) &= (W^{\lambda}_{x, n})^* \hA^{\lambda, x, n} W^{\lambda}_{x, n} \pi^{\lambda}_{n}(P) = (W^{\lambda}_{x, n})^* \hA^{\lambda, x, n} \pi^{\lambda}_{x, n}(P) W^{\lambda}_{x, n} \\
            &= (W^{\lambda}_{x, n})^* \pi^{\lambda}_{x, n}(P) \hA^{\lambda, x, n} W^{\lambda}_{x, n} =  \pi^{\lambda}_{n}(P) (W^{\lambda}_{x, n})^* \hA^{\lambda, x, n} W^{\lambda}_{x, n} = \pi^{\lambda}_{n}(P) \hA^{\lambda, n}.
        \end{aligned}
    \end{equation}
    For the positivity claim in statement \textit{(i)}, with any $P \in \cB$ we can check that
    \begin{equation}\label{eq:t1-GNS-support-7}
        \begin{aligned}
            &\sandwich{\Omega^{\lambda}_{n}}{\pi^{\lambda}_{n}(P)^* \cdot \hA^{\lambda, n} \cdot \pi^{\lambda}_{n}(P)}{\Omega^{\lambda}_{n}} \\
            &\quad = \sandwich{\Omega^{\lambda}_{n}}{\pi^{\lambda}_{n}(P)^* (W^{\lambda}_{x, n})^* \cdot \hA^{\lambda, x, n} \cdot W^{\lambda}_{x, n} \pi^{\lambda}_{n}(P)}{\Omega^{\lambda}_{n}} \\
            &\quad = \sandwich{\Omega^{\lambda}_{x, n}}{\pi^{\lambda}_{x, n}(P)^* \cdot \hA^{\lambda, x, n} \cdot \pi^{\lambda}_{x, n}(P)}{\Omega^{\lambda}_{x, n}} \geq 0
        \end{aligned}
    \end{equation}
    by positivity of $\hA^{\lambda, x, n} \in B(\cH^{\lambda}_{x, n})$.

    Then, we prove the almost-POVM condition in statement \textit{(ii)}.
    Since $\{\hA^{\lambda,x,n}\}_a$ is a POVM,
    \begin{align}\label{eq:t1-sum-A-W}
        \sum_a\hA^{\lambda,n} =(W^\lambda_{x,n})^*\left(\sum_a\hA^{\lambda,x,n}\right)W^\lambda_{x,n} =(W^\lambda_{x,n})^*W^\lambda_{x,n}.
    \end{align}
    Consequently,
    \begin{equation}\label{eq:t1-product-coeff-bound}
        \begin{aligned}
            &\left|\sandwich{\Omega^\lambda_n}{\pi^\lambda_n(P_1) \left(\sum_a\hA^{\lambda,n}-\id\right)\pi^\lambda_n(P_2)}{\Omega^\lambda_n}\right|\\
            &\quad=|(\widetilde{\sigma}^{\lambda,n}_x-\widetilde{\sigma}^{\lambda,n})(P_1P_2)|\\
            &\quad\leq\frac{\eta_n(\lambda)}{1+\eta'_n(\lambda)}\coeffnorm{P_1}\coeffnorm{P_2},
        \end{aligned}
    \end{equation}
    which proves \cref{eq:AAlmostPOVM}.

    Fianlly, to prove statement \textit{(iii)}, first note that if $\eta_n(\lambda)=0$, the marginal differences vanish on $\cB_{2n}$ by \cref{eq:t1-average-coeff-bound}, so the result holds.
    Thus, assume that $\eta_n(\lambda)>0$.
    We first show that the map $P\mapsto\pi^\lambda_n(P)\ket{\Omega^\lambda_n}$ preserves the dimension of $\cB_n$.
    By \cref{eq:Reviseregularization} and positivity of each $\sigma^{\lambda,n}_{a|x}$, for every positive $R\in\cB$ we have
    \begin{equation}\label{eq:t1-reference-lower-bound}
        \begin{aligned}
            \widetilde{\sigma}^{\lambda,n}(R)
            =\frac{|I_X|^{-1}\sum_{x,a}\sigma^{\lambda,n}_{a|x}(R)+\eta'_n(\lambda)\rho(R)}{1+\eta'_n(\lambda)} \geq\frac{\eta'_n(\lambda)}{1+\eta'_n(\lambda)}\rho(R).
        \end{aligned}
    \end{equation}
    Faithfulness of $\rho$ therefore gives, for every nonzero $P\in\cB$,
    \begin{align}\label{eq:t1-GNS-injective}
        \bigl\|\pi^\lambda_n(P)\ket{\Omega^\lambda_n}\bigr\|^2=\widetilde{\sigma}^{\lambda,n}(P^*P)\geq\frac{\eta'_n(\lambda)}{1+\eta'_n(\lambda)}\rho(P^*P)>0.
    \end{align}
    Thus this linear map is injective, and its image on $\cB_n$ is $V_n$, so $\dim(V_n)=\dim(\cB_n)$.

    The same calculation works for all the tests, $R=P^* P, Q^* \B Q, R = Q^* (\id - \B) Q$, in statement \textit{(iii)}.
    Choose a fixed basis in each of the spaces $\cB_n$, $B_{b|y}^{1/2}\cB_{n-1}$, and $(\id-B_{b|y})^{1/2}\cB_{n-1}$, for every $(b,y)\in I_B\times I_Y$, orthonormal for the inner product $(P,Q)\mapsto\rho(P^*Q)$.
    Such bases exist because $\rho$ is faithful and with Gram--Schmidt construction.
    (Note that $B_{b|y}^{1/2}$ and $(\id-B_{b|y})$ exists as positive elements on the full $C^*$-algebra $\cB$, but $B_{b|y}^{1/2}\cB_{n-1}$, and $(\id-B_{b|y})^{1/2}\cB_{n-1}$ need not lie in $\cB_{n}$.
    This however does not affect the subsequent arguments).
    
    Each of $\cB_n$, $B_{b|y}^{1/2}\cB_{n-1}$, and $(\id-B_{b|y})^{1/2}\cB_{n-1}$ has dimension at most $\dim(V_n)$.
    Indeed, the first has dimension $\dim(\cB_n)$, and each of the other spaces is a linear image of $\cB_{n-1} \subset \cB_n$.
    Write $\{P_i\}$ for any one of these bases, so $\rho(P_i^*P_j)=\delta_{ij}$.
    Every product $P_i^*P_j$ belongs to $\cB_{2n}$: for the second space it has the form $Q_i^*B_{b|y}Q_j$ with $Q_i,Q_j\in\cB_{n-1}$, and for the third it has the form $Q_i^*(\id-B_{b|y})Q_j$.
    Choose $C_n>0$ large enough so that it bounds all $\coeffnorm{ P_i^* P_j }$ of these chosen bases.
    The constant $C_n$ is independent of $\lambda$ and $S$ since the definition depends only on $\rho$ and $\cB$.

    Every test $R$ in statement \textit{(iii)} is $P^*P$ for some $P$ in one of these spaces.
    Expand $P=\sum_i c_iP_i$ in its corresponding basis, so $\rho(P^*P)=\sum_i|c_i|^2$ by orthonormality.
    Using \cref{eq:t1-average-coeff-bound,eq:t1-reference-lower-bound}, we obtain
    \begin{equation}\label{eq:t1-relative-proof}
        \begin{aligned}
            |(\widetilde{\sigma}^{\lambda,n}_x-\widetilde{\sigma}^{\lambda,n})(P^*P)|
            &=\left|\sum_{i,j}\overline{c_i}c_j(\widetilde{\sigma}^{\lambda,n}_x-\widetilde{\sigma}^{\lambda,n})(P_i^*P_j)\right| \leq\frac{C_n\eta_n(\lambda)}{1+\eta'_n(\lambda)}\left(\sum_i|c_i|\right)^2\\
            &\leq\frac{\dim(V_n)C_n\eta_n(\lambda)}{1+\eta'_n(\lambda)}\rho(P^*P) \leq\dim(V_n)C_n\eta'_n(\lambda)\widetilde{\sigma}^{\lambda,n}(P^*P).
        \end{aligned}
    \end{equation}
    The first inequality uses the triangle inequality, orthonormality, and $\coeffnorm{P_i^*P_j}\leq C_n$.
    For the second, Cauchy--Schwarz gives $(\sum_i|c_i|)^2\leq(\sum_i1)\sum_i|c_i|^2\leq\dim(V_n)\rho(P^*P)$, since the chosen basis has at most $\dim(V_n)$ elements.
    The last inequality substitutes $\rho(P^*P)\leq(1+\eta'_n)\widetilde{\sigma}^{\lambda,n}(P^*P)/\eta'_n$ and uses $\eta_n/\eta'_n= \eta_n/\sqrt{\eta_n} = \eta'_n$.
    With $\etaLem_n:=C_n\eta'_n$, this proves \cref{eq:t1-relative-tests}.
\end{proof}


In particular, if $\sigma^{\lambda,n}_x(P)=\sigma^{\lambda,n}_{x'}(P)$ for all $x,x'\in I_X$ and $P\in\cB_{2n}$, we may take $\eta_n(\lambda)=\eta'_n(\lambda)=0$ in \cref{prop:GNSConstructionCommonSpace}.
Then $\widetilde{\sigma}^{\lambda,n}_{a|x}=\sigma^{\lambda,n}_{a|x}$ and $\widetilde p^{\lambda,n}=p^\lambda$, so the representation reproduces the original compiled correlation without perturbing the functionals.
This is, for example, the case in the asymptotic limit $\lambda \to \infty$ for each fixed $n$.

Let $\|\vec{\beta}\|_1:=\sum_{a,b,x,y}|\beta_{abxy}|$.
With the quantum representation constructed by \cref{prop:GNSConstructionCommonSpace}, the compiled Bell score for $\cG_{\comp}$ with QPT strategy $S$ satisfies, by \cref{eq:t1-correlation-error},
\begin{equation}\label{eq:BellValueFiniteLambdaFiniteNPA}
    \begin{aligned}
        \gamevalueCompile{\lambda}{\cG_{\comp},S}
        &=\langle p^\lambda,\vec{\beta}\rangle
        \leq\langle\widetilde p^{\lambda,n},\vec{\beta}\rangle+\|\vec{\beta}\|_1\eta'_n(\lambda)\\
        &=\sandwich{\Omega^\lambda_n}{\beta(\hA^{\lambda,n},\pi^\lambda_n(\B))}{\Omega^\lambda_n}+\|\vec{\beta}\|_1\eta'_n(\lambda),
    \end{aligned}
\end{equation}
where the correction $\|\vec{\beta}\|_1\eta'_n(\lambda)$ is negligible.

The goal now is to relate $\sandwich{\Omega^\lambda_n}{\beta(\hA^{\lambda,n},\pi^\lambda_n(\B))}{\Omega^\lambda_n}$ to the $n$-th level sequential NPA hierarchy \cref{eq:SequentialNPASDP}.
The constructed representation, however, need not define a feasible solution, since the non-signaling constraints hold only approximately, as expressed in \cref{eq:AAlmostPOVM}.
We address this remaining signaling through the signaling/non-signaling decomposition in \cref{prop:DecompositionNSSI}.

\subsection{Signaling/non-signaling decompositions}\label{sec:TechnicalDecomposition}
Following the observation above, it is important to quantify the signaling effect on the compiled Bell score in order to identify with the sequential NPA hierarchy.
Therefore, this section contains the main technical result (\cref{prop:DecompositionNSSI}) inspired by the approach in~\cite{renou2017inequivalence}: using group representation theory, we are able to identify the parts of $\hA^{\lambda,n}$ that are no-signaling and signaling, and consequently bound the advantage of signaling with negligible functions.

We begin with the observation that $\sum_a\hA^{\lambda,n}$ can be dominated by $\id_{\cH^\lambda_n}$ on the low-degree subspace upon rescaling thanks to \cref{prop:GNSConstructionCommonSpace}.
Specifically, by \cref{eq:t1-relative-tests}, one has
\begin{equation}\label{eq:sumAaxdominated}
    \begin{aligned}
        &\sandwich{\Omega^\lambda_n}{\pi^\lambda_n(P)^*\left(\id_{\cH^\lambda_n}-\frac{\sum_a\hA^{\lambda,n}}{1+\dim(V_n)\etaLem_n(\lambda)}\right)\pi^\lambda_n(P)}{\Omega^\lambda_n} = \widetilde{\sigma}^{\lambda,n}(P^*P)-\frac{\widetilde{\sigma}^{\lambda,n}_x(P^*P)}{1+\dim(V_n)\etaLem_n(\lambda)}\geq0.
    \end{aligned}
\end{equation}
for every $P\in\cB_n$, where $V_n$ and $\etaLem_n$ are as in \cref{prop:GNSConstructionCommonSpace}.
In other words, the operator $\id_{\cH^\lambda_n}-(1+\dim(V_n)\etaLem_n(\lambda))^{-1}\sum_a\hA^{\lambda,n}$ remains positive semidefinite on the low-degree subspace $V_n$.

The following proposition provides a systematic method for decomposing the measurement operators $\hA^{\lambda,n}$ into three parts: a no-signaling component $\hA^{\lambda,n}(\mathrm{NS})$, a signaling component $\hA^{\lambda,n}(\mathrm{SI})$, and a residue component $\hA^{\lambda,n}(\mathrm{res})$ that ensures overall physicality (i.e., positivity).
This decomposition is not only central to the discussion in \cref{sec:MainTheorems}, but may also offer interesting insights into related questions, such as the role of signaling effects in quantum steering.

\begin{proposition}\label{prop:DecompositionNSSI}
    Consider the quantum representation as constructed in \cref{prop:GNSConstructionCommonSpace} for a QPT strategy $S$ of a compiled Bell game $\cG_{\comp}$ with respect to NPA level $n$.
    Then, there exists a decomposition
    \begin{align}\label{eq:DecompositionNSSIRes}
        \hA^{\lambda,n}=\hA^{\lambda,n}(\mathrm{NS})+\hA^{\lambda,n}(\mathrm{SI})+\frac{\dim(V_n)\etaLem_n(\lambda)}{1+\dim(V_n)\etaLem_n(\lambda)}\hA^{\lambda,n}(\mathrm{res}),
    \end{align}
    where $V_n=\operatorname{span}\{\pi^\lambda_n(P)\ket{\Omega^\lambda_n}\mid P\in\cB_n\}$ and $\etaLem_n(\lambda)$ is the same negligible function constructed in \cref{prop:GNSConstructionCommonSpace}.
    Furthermore,
    \begin{enumerate}[label=(\roman*)]
        \item $\hA^{\lambda,n}(\mathrm{NS}),\hA^{\lambda,n}(\mathrm{SI}),\hA^{\lambda,n}(\mathrm{res})\in\pi^\lambda_n(\cB)'\subset B(\cH^\lambda_n)$, i.e., commutativity is preserved with the decomposition.
        \item For every $P\in\cB_{2n}$,
        \begin{align}\label{eq:DecompositionSignalingCoeffBound}
            \left|\sandwich{\Omega^\lambda_n}{\hA^{\lambda,n}(\mathrm{SI})\pi^\lambda_n(P)}{\Omega^\lambda_n}\right|\leq\frac{C_G\eta_n(\lambda)}{1+\eta'_n(\lambda)}\coeffnorm{P},
        \end{align}
        where $C_G>0$ depends only on $I_A$ and $I_X$.
        That is, the contribution from the signaling effect from Alice to Bob is negligible for every fixed low-degree polynomial.
        \item For any $P_1,P_2\in\cB_n$,
        \begin{align}\label{eq:DecompositionNSNormalization}
            \sandwich{\Omega^\lambda_n}{\pi^\lambda_n(P_1)\left(\sum_a\hA^{\lambda,n}(\mathrm{NS})\right)\pi^\lambda_n(P_2)}{\Omega^\lambda_n}=\sandwich{\Omega^\lambda_n}{\pi^\lambda_n(P_1)\id_{\cH^\lambda_n}\pi^\lambda_n(P_2)}{\Omega^\lambda_n},
        \end{align}
        i.e., no-signaling on the low-degree polynomial subspace.
        \item For any $P\in\cB_n$,
        \begin{align}\label{eq:DecompositionNSPositive}
            \sandwich{\Omega^\lambda_n}{\pi^\lambda_n(P)^*\hA^{\lambda,n}(\mathrm{NS})\pi^\lambda_n(P)}{\Omega^\lambda_n}\geq0,
        \end{align}
        i.e., $\hA^{\lambda,n}(\mathrm{NS})$ is positive on the low-degree polynomial subspace.
        In addition,
        \begin{align}\label{eq:DecompositionNSLocalizing}
            \sandwich{\Omega^\lambda_n}{\hA^{\lambda,n}(\mathrm{NS})\pi^\lambda_n(R)}{\Omega^\lambda_n}\geq0
        \end{align}
        for $R=Q^*B_{b|y}Q$ or $R=Q^*(\id-B_{b|y})Q$, with $Q\in\cB_{n-1}$ and $(b,y)\in I_B\times I_Y$.
    \end{enumerate}
    Observe from (iii) and (iv) that $\hA^{\lambda,n}(\mathrm{NS})$ satisfies POVM conditions on the low-degree subspace $V_n$, together with the corresponding POVM localizing inequalities.
\end{proposition}

\begin{proof}
    Let us use physical intuition to identify the signaling part of $\hA^{\lambda,n}$ from Alice to Bob.
    To Bob, all he can see from Alice is the effect of the marginal $\sum_a\hA^{\lambda,n}$, or equivalently the average over the symbol $a$.
    Suppose that there is no signaling at all, then to Bob the marginal $\sum_a\hA^{\lambda,n}$ should be $x$-label invariant.
    Consequently, the complement of the $x$-invariant part of $\sum_a\hA^{\lambda,n}$---the part that is sensitive to any change in $x$---represents the signaling effect from Alice to Bob.
    It turns out the symmetric group and its representation theory are the best for describing our physical intuition, which we adapt in our proof.

    \proofstep{1}{Notation of symmetry group representation and Young symmetrizers:}
    Let the symmetric group $S_{|I_A|}$ act on $\hA^{\lambda,n}$ by permuting the $a$ index, $s:\hA^{\lambda,n}\mapsto\hat{A}^{\lambda,n}_{s(a)|x}$.
    (Note that they are merely symbolic actions on $\hA^{\lambda,n}$ rather than a full action on $B(\cH^\lambda_n)$.)
    Denote by $\Pi^a_\mu$ the normalized Young symmetrizer associated with a standard Young tableau $\mu$, with $\mu=0$ denoting the trivial tableau, and define
    \begin{equation}\label{eq:DecompositionOutcomeProjections}
        \begin{aligned}
            \Pi^a_0&:=\Pi^a_{\mu=0},\\
            \Pi^a_1&:=\sum_{\mu\neq0}\Pi^a_\mu.
        \end{aligned}
    \end{equation}
    Then $\Pi^a_0(\hA^{\lambda,n})$ is precisely the average over symbols $a$ (i.e., the marginal), while $S_{|I_A|}$ acts non-trivially on $\Pi^a_1(\hA^{\lambda,n})$, such that $\Pi^a_0(\hA^{\lambda,n})+\Pi^a_1(\hA^{\lambda,n})=\hA^{\lambda,n}$.
    Also, they are mutually orthogonal in the sense that $\Pi^a_0\Pi^a_1(\hA^{\lambda,n})=\Pi^a_1\Pi^a_0(\hA^{\lambda,n})=0$.

    Analogously, consider the symmetric group $S_{|I_X|}$ acting on $\hA^{\lambda,n}$ by permuting the $x$ index, $s:\hA^{\lambda,n}\mapsto\hat{A}^{\lambda,n}_{a|s(x)}$.
    We similarly denote by $\Pi^x_\mu$ the normalized Young symmetrizers for standard Young tableaux and define
    \begin{equation}\label{eq:DecompositionInputProjections}
        \begin{aligned}
            \Pi^x_0&:=\Pi^x_{\mu=0},\\
            \Pi^x_1&:=\sum_{\mu\neq0}\Pi^x_\mu.
        \end{aligned}
    \end{equation}
    We also have that $\Pi^x_0(\hA^{\lambda,n})+\Pi^x_1(\hA^{\lambda,n})=\hA^{\lambda,n}$ and $\Pi^x_0\Pi^x_1(\hA^{\lambda,n})=\Pi^x_1\Pi^x_0(\hA^{\lambda,n})=0$.
    It is clear from the definition that the action of $\Pi^a_i$ commutes with $\Pi^x_j$ on $\hA^{\lambda,n}$, so we can unambiguously apply them jointly.

    \proofstep{2}{Identifying the signaling contribution}
    Following from the above remark, the signaling part then corresponds to the marginal of Bob, i.e., $\Pi^a_0$, that is purely non-invariant under permutation of $x$, i.e., $\Pi^x_1$.
    Thus we define the signaling contribution by
    \begin{align}\label{eq:DecompositionSIDefinition}
        \hA^{\lambda,n}(\mathrm{SI}):=\Pi^a_0\Pi^x_1(\hA^{\lambda,n}),
    \end{align}
    which lies in $\pi^\lambda_n(\cB)'$ as it is a linear combination of $\hA^{\lambda,n}$.

    \proofstep{3}{Checking (ii) bound on signaling part for low-degrees:}
    For any nontrivial standard Young tableau $\mu$, the associated normalized Young symmetrizer $\Pi^x_\mu$ is a constant multiple of the difference of two equally-sized sums of permutations, each having at most ${|I_X|!/2}$ terms~\cite{procesi2007lie}.
    Consequently, when applied to $\sum_a\hA^{\lambda,n}$, the expression $\sandwich{\Omega^\lambda_n}{\Pi^x_\mu(\sum_a\hA^{\lambda,n})\pi^\lambda_n(P)}{\Omega^\lambda_n}$ is a fixed linear combination of differences of the form
    \begin{align}\label{eq:DecompositionMarginalDifference}
        \sum_a\sandwich{\Omega^\lambda_n}{\left(\hat{A}^{\lambda,n}_{a|x'}-\hat{A}^{\lambda,n}_{a|x''}\right)\pi^\lambda_n(P)}{\Omega^\lambda_n}=\widetilde{\sigma}^{\lambda,n}_{x'}(P)-\widetilde{\sigma}^{\lambda,n}_{x''}(P).
    \end{align}
    Thus, for any $P\in\cB_{2n}$,
    \begin{equation}\label{eq:DecompositionSignalingProof}
        \begin{aligned}
            \left|\sandwich{\Omega^\lambda_n}{\hA^{\lambda,n}(\mathrm{SI})\pi^\lambda_n(P)}{\Omega^\lambda_n}\right| &\quad=\left|\sandwich{\Omega^\lambda_n}{\Pi^a_0\Pi^x_1(\hA^{\lambda,n})\pi^\lambda_n(P)}{\Omega^\lambda_n}\right|\\
            &\quad=\left|\frac{1}{|I_A|}\sum_{\mu\neq0}\sandwich{\Omega^\lambda_n}{\left(\Pi^x_\mu(\sum_a\hA^{\lambda,n})\right)\pi^\lambda_n(P)}{\Omega^\lambda_n}\right|\\
            &\quad\leq C_G\max_{x',x''\in I_X}\left|\widetilde{\sigma}^{\lambda,n}_{x'}(P)-\widetilde{\sigma}^{\lambda,n}_{x''}(P)\right| \leq\frac{C_G\eta_n(\lambda)}{1+\eta'_n(\lambda)}\coeffnorm{P},
        \end{aligned}
    \end{equation}
    for a constant $C_G>0$ depending only on $I_A$ and $I_X$, including the normalization factors of the Young symmetrizers.
    The last inequality follows from \cref{eq:ReviseDegnNegPolyCoeff,eq:Reviseregularization} and that the $\rho$-terms cancel out.

    \proofstep{4}{Constructing the no-signaling and the residual part:}
    It remains to identify $\hA^{\lambda,n}(\mathrm{NS})$, the component that appears to be POVM on the low-degree subspace $\cB_n$.
    One natural choice is the complement of the signaling contribution, i.e.,
    \begin{align}\label{eq:DecompositionUnscaledNS}
        \hA^{\lambda,n}-\Pi^a_0\Pi^x_1(\hA^{\lambda,n})=\Pi^a_0\Pi^x_0(\hA^{\lambda,n})+\Pi^a_1(\hA^{\lambda,n}).
    \end{align}
    However, while it satisfies \emph{(i), (iii)}, it may fail condition \emph{(iv)} due to the fact that $\Pi^a_1(\hA^{\lambda,n})$ can be negative.
    Therefore, the correct definition is by rescaling $\Pi^a_1(\hA^{\lambda,n})$ to make it less harmful to the overall positivity.
    Thanks to \cref{eq:sumAaxdominated}, we already have a candidate for the scaling factor and may define
    \begin{align}\label{eq:DecompositionNSDefinition}
        \hA^{\lambda,n}(\mathrm{NS}):=\Pi^a_0\Pi^x_0(\hA^{\lambda,n})+\frac{1}{1+\dim(V_n)\etaLem_n(\lambda)}\Pi^a_1(\hA^{\lambda,n}).
    \end{align}
    Consequently, the residual part is simply
    \begin{align}\label{eq:DecompositionResidualDefinition}
        \hA^{\lambda,n}(\mathrm{res}):=\Pi^a_1(\hA^{\lambda,n}),
    \end{align}
    so that \cref{eq:DecompositionNSSIRes} holds.

    \proofstep{5}{Verifying (iii) the low-degree no-signaling:}
    To this end, observe that $\sum_a\Pi^a_1(\hA^{\lambda,n})=|I_A|\Pi^a_0\Pi^a_1(\hA^{\lambda,n})=0$ and $\sum_a\Pi^a_0(\hA^{\lambda,n})=|I_A|\Pi^a_0(\hA^{\lambda,n})$.
    So for any $P_1,P_2\in\cB_n$ we have
    \begin{equation}\label{eq:DecompositionNSNormalizationProof}
        \begin{aligned}
            &\sandwich{\Omega^\lambda_n}{\pi^\lambda_n(P_1)\left(\sum_a\hA^{\lambda,n}(\mathrm{NS})\right)\pi^\lambda_n(P_2)}{\Omega^\lambda_n}\\
            &\quad=\sandwich{\Omega^\lambda_n}{\pi^\lambda_n(P_1)\left(\sum_a\Pi^a_0\Pi^x_0(\hA^{\lambda,n})\right)\pi^\lambda_n(P_2)}{\Omega^\lambda_n}\\
            &\quad=|I_A|\frac{1}{|I_A||I_X|}\sum_{a,x}\sandwich{\Omega^\lambda_n}{\pi^\lambda_n(P_1)\hA^{\lambda,n}\pi^\lambda_n(P_2)}{\Omega^\lambda_n}\\
            &\quad=\frac{1}{|I_X|}\sum_{a,x}\widetilde{\sigma}^{\lambda,n}_{a|x}(P_1P_2)=\widetilde{\sigma}^{\lambda,n}(P_1P_2)=\sandwich{\Omega^\lambda_n}{\pi^\lambda_n(P_1)\id_{\cH^\lambda_n}\pi^\lambda_n(P_2)}{\Omega^\lambda_n},
        \end{aligned}
    \end{equation}
    as desired.
    Observe that the above calculation also shows that $\Pi^a_0\Pi^x_0(\hA^{\lambda,n})$ is the same as $\frac{1}{|I_A|}\id_{\cH^\lambda_n}$ in the low-degree subspace, which will be useful for the next step.

     \proofstep{6}{Checking (iv) positivity on low-degrees:}
    Note that
    \begin{equation}\label{eq:DecompositionNSPositiveSplit}
        \begin{aligned}
            \hA^{\lambda,n}(\mathrm{NS})
            &=\Pi^a_0\Pi^x_0(\hA^{\lambda,n})+\frac{1}{1+\dim(V_n)\etaLem_n(\lambda)}\Pi^a_1(\hA^{\lambda,n})\\
            &=\frac{1}{1+\dim(V_n)\etaLem_n(\lambda)}\hA^{\lambda,n} + \left(\Pi^a_0\Pi^x_0(\hA^{\lambda,n})-\frac{1}{1+\dim(V_n)\etaLem_n(\lambda)}\Pi^a_0(\hA^{\lambda,n})\right).
        \end{aligned}
    \end{equation}
    Hence, it follows from \cref{eq:sumAaxdominated}, the positivity of $\hA^{\lambda,n}$, and the final observation of \emph{Step 5} that
    \begin{equation}\label{eq:DecompositionNSPositiveProof}
        \begin{aligned}
            &\sandwich{\Omega^\lambda_n}{\pi^\lambda_n(P)^*\hA^{\lambda,n}(\mathrm{NS})\pi^\lambda_n(P)}{\Omega^\lambda_n}\\
            &\quad=\frac{1}{1+\dim(V_n)\etaLem_n(\lambda)}\sandwich{\Omega^\lambda_n}{\pi^\lambda_n(P)^*\hA^{\lambda,n}\pi^\lambda_n(P)}{\Omega^\lambda_n}\\
            &\qquad+\sandwich{\Omega^\lambda_n}{\pi^\lambda_n(P)^*\left(\Pi^a_0\Pi^x_0(\hA^{\lambda,n})-\frac{\Pi^a_0(\hA^{\lambda,n})}{1+\dim(V_n)\etaLem_n(\lambda)}\right)\pi^\lambda_n(P)}{\Omega^\lambda_n}\\
            &\quad\geq\frac{1}{|I_A|}\sandwich{\Omega^\lambda_n}{\pi^\lambda_n(P)^*\left(\id_{\cH^\lambda_n}-\frac{\sum_a\hA^{\lambda,n}}{1+\dim(V_n)\etaLem_n(\lambda)}\right)\pi^\lambda_n(P)}{\Omega^\lambda_n}\geq0
        \end{aligned}
    \end{equation}
    for every $P\in\cB_n$.

    For the POVM localizing inequalities, let $R=Q^*B_{b|y}Q$ or $R=Q^*(\id-B_{b|y})Q$, with $Q\in\cB_{n-1}$.
    By the definitions of the averaging projections and the GNS identities,
    \begin{equation}\label{eq:DecompositionNSLocalizingProof}
        \begin{aligned}
            \sandwich{\Omega^\lambda_n}{\hA^{\lambda,n}(\mathrm{NS})\pi^\lambda_n(R)}{\Omega^\lambda_n}
            &=\frac{\widetilde{\sigma}^{\lambda,n}_{a|x}(R)}{1+\dim(V_n)\etaLem_n(\lambda)}\\
            &\quad+\frac{1}{|I_A|}\left(\widetilde{\sigma}^{\lambda,n}(R)-\frac{\widetilde{\sigma}^{\lambda,n}_x(R)}{1+\dim(V_n)\etaLem_n(\lambda)}\right)\geq0.
        \end{aligned}
    \end{equation}
    The first term is nonnegative by positivity of $\widetilde{\sigma}^{\lambda,n}_{a|x}$, and the second is nonnegative by \cref{eq:t1-relative-tests}.
\end{proof}

\subsection{Quantitative characterization of compiled Bell games}\label{sec:MainTheorems}
The decomposition \cref{prop:DecompositionNSSI} gives rise to $\hA^{\lambda,n}(\mathrm{NS})$, $\hA^{\lambda,n}(\mathrm{SI})$, and $\hA^{\lambda,n}(\mathrm{res})$.
Let us analyze each of them individually, for the fixed QPT strategy $S$ and hierarchy level $n$.
\begin{enumerate}
    \item First, \emph{(iii), (iv)} of \cref{prop:DecompositionNSSI} imply that the linear functionals
    \begin{align}\label{eq:NSFunctionalDefinition}
        \widetilde{\sigma}^{\lambda,n,\mathrm{NS}}_{a|x}(P):=\sandwich{\Omega^\lambda_n}{\hA^{\lambda,n}(\mathrm{NS})\pi^\lambda_n(P)}{\Omega^\lambda_n},\qquad P\in\cB_{2n},
    \end{align}
    is positive, along with the positive semidefinite localizing matrices associated with $\B$ and $\id - \B$.
    They also satisfy $\sum_a\widetilde{\sigma}^{\lambda,n,\mathrm{NS}}_{a|x}(P)=\widetilde{\sigma}^{\lambda,n}(P)$ for all $P\in\cB_{2n}$.
    Since $\widetilde{\sigma}^{\lambda,n}(\id)=1$, these are precisely the positivity, strongly no-signaling, and normalization conditions in \cref{eq:SequentialNPASDP}; Bob's completeness constraints follow from $\sum_b B_{b|y}=\id$ in $\cB$.
    Consequently, the correlation
    \begin{align}\label{eq:NSCorrelationDefinition}
        \widetilde p^{\lambda,n}_{\mathrm{NS}}(ab|xy):=\widetilde{\sigma}^{\lambda,n,\mathrm{NS}}_{a|x}(\B)
    \end{align}
    is compatible with the $n$-th level of the strongly no-signaling sequential NPA hierarchy.

    Thus, the corresponding Bell score for $\widetilde p^{\lambda,n}_{\mathrm{NS}}$ is upper-bounded by the optimal sequential NPA score at level $n$:
    \begin{align}\label{eq:NSScoreBound}
        \omega^{\lambda,n}_{\mathrm{NS}}:=\langle\widetilde p^{\lambda,n}_{\mathrm{NS}},\vec{\beta}\rangle\leq\gamevalueSeqNPA{\cG}{n}.
    \end{align}

    \item Next, consider the $n$-dependent pseudo-correlations (due to potential negativity)
    \begin{align}\label{eq:SICorrelationDefinition}
        \widetilde p^{\lambda,n}_{\mathrm{SI}}(ab|xy):=\sandwich{\Omega^\lambda_n}{\hA^{\lambda,n}(\mathrm{SI})\pi^\lambda_n(\B)}{\Omega^\lambda_n}.
    \end{align}
    Since $\coeffnorm{\B}=1$, \cref{eq:DecompositionSignalingCoeffBound} gives
    \begin{align}\label{eq:SICorrelationBound}
        |\widetilde p^{\lambda,n}_{\mathrm{SI}}(ab|xy)|\leq\frac{C_G\eta_n(\lambda)}{1+\eta'_n(\lambda)}
    \end{align}
    for every $a,b,x,y$.
    In particular, the score contribution $\omega^{\lambda,n}_{\mathrm{SI}}:=\langle\widetilde p^{\lambda,n}_{\mathrm{SI}},\vec{\beta}\rangle$ satisfies
    \begin{equation}\label{eq:SIScoreBound}
        \begin{aligned}
            |\omega^{\lambda,n}_{\mathrm{SI}}|
            &\leq\sum_{a,b,x,y}|\beta_{abxy}|\,|\widetilde p^{\lambda,n}_{\mathrm{SI}}(ab|xy)| \leq\frac{C_G\|\vec{\beta}\|_1\eta_n(\lambda)}{1+\eta'_n(\lambda)}.
        \end{aligned}
    \end{equation}

    \item Lastly, consider the $n$-dependent pseudo-correlation
    \begin{align}\label{eq:ResidualCorrelationDefinition}
        \widetilde p^{\lambda,n}_{\mathrm{res}}(ab|xy):=\sandwich{\Omega^\lambda_n}{\hA^{\lambda,n}(\mathrm{res})\pi^\lambda_n(\B)}{\Omega^\lambda_n}.
    \end{align}
    By \cref{eq:DecompositionResidualDefinition} and the definition of $\Pi^a_1$, we have
    \begin{align}\label{eq:ResidualCorrelationFormula}
        \widetilde p^{\lambda,n}_{\mathrm{res}}(ab|xy)=\widetilde p^{\lambda,n}(ab|xy)-\frac{1}{|I_A|}\sum_{a'}\widetilde p^{\lambda,n}(a'b|xy).
    \end{align}
    Both terms on the right belong to $[0,1]$, so $|\widetilde p^{\lambda,n}_{\mathrm{res}}(ab|xy)|\leq1$.
    Thus its score contribution $\omega^{\lambda,n}_{\mathrm{res}}:=\langle\widetilde p^{\lambda,n}_{\mathrm{res}},\vec{\beta}\rangle$ satisfies
    \begin{align}\label{eq:ResidualScoreBound}
        |\omega^{\lambda,n}_{\mathrm{res}}|\leq\sum_{a,b,x,y}|\beta_{abxy}|\,|\widetilde p^{\lambda,n}_{\mathrm{res}}(ab|xy)|\leq\|\vec{\beta}\|_1.
    \end{align}
\end{enumerate}

With the above decomposition, we have already done most of the proof for the following main result, which upper-bounds the compiled Bell score with the sequential NPA hierarchy value $\gamevalueSeqNPA{\cG}{n}$ and an NPA level dependent negligible function $\etaThm_{S,n}(\lambda)$.
\begin{theorem}\label{thm:GeneralQuantBoundFiniteSecurityQCValue}
    Let $\cG$ be a bipartite Bell game.
    Consider its compiled version $\cG_{\comp}$ and let $S=(S_\lambda)_\lambda$ be an arbitrary quantum polynomial time (QPT) strategy employed by the prover.
    Let the approximation error of the sequential NPA hierarchy for $\cG$ be $\epsilon(n):=\gamevalueSeqNPA{\cG}{n}-\gamevalueqcopt{\cG}$, where $\epsilon(n)\to0$ monotonically as $n\to\infty$.

    Then, for every fixed integer $n\geq1$, there exists a negligible function $\etaThm_{S,n}(\lambda)$, depending on the QHE scheme, the strategy $S$, and $n$, such that
    \begin{align}\label{eqthm}
        \gamevalueCompile{\lambda}{\cG_{\comp},S}\leq\gamevalueSeqNPA{\cG}{n}+\etaThm_{S,n}(\lambda)=\gamevalueqcopt{\cG}+\epsilon(n)+\etaThm_{S,n}(\lambda),
    \end{align}
    where $\gamevalueCompile{\lambda}{\cG_{\comp},S}$ is the prover's Bell score using the QPT strategy $S$.
    In other words, the Bell score derived from the QPT strategy $S$ via NPA level $n$ analysis is upper-bounded by the optimal score of the sequential NPA hierarchy at level $n$ plus $\etaThm_{S,n}(\lambda)$.
\end{theorem}
\begin{proof}
    Thanks to the discussion preceding the theorem, \cref{eq:BellValueFiniteLambdaFiniteNPA,eq:DecompositionNSSIRes} give
    \begin{equation}\label{eq:CompiledScoreDecompositionBound}
        \begin{aligned}
            \gamevalueCompile{\lambda}{\cG_{\comp},S}
            &\leq\langle\widetilde p^{\lambda,n},\vec{\beta}\rangle+\|\vec{\beta}\|_1\eta'_n(\lambda)\\
            &=\omega^{\lambda,n}_{\mathrm{NS}}+\omega^{\lambda,n}_{\mathrm{SI}} +\frac{\dim(V_n)\etaLem_n(\lambda)}{1+\dim(V_n)\etaLem_n(\lambda)}\omega^{\lambda,n}_{\mathrm{res}}+\|\vec{\beta}\|_1\eta'_n(\lambda)\\
            &\leq\gamevalueSeqNPA{\cG}{n}+|\omega^{\lambda,n}_{\mathrm{SI}}| +\frac{\dim(V_n)\etaLem_n(\lambda)}{1+\dim(V_n)\etaLem_n(\lambda)}|\omega^{\lambda,n}_{\mathrm{res}}|+\|\vec{\beta}\|_1\eta'_n(\lambda)\\
            &\leq\gamevalueSeqNPA{\cG}{n}+\|\vec{\beta}\|_1\left(\frac{C_G\eta_n(\lambda)}{1+\eta'_n(\lambda)}+\dim(V_n)\etaLem_n(\lambda)+\eta'_n(\lambda)\right)\\
            &=\gamevalueSeqNPA{\cG}{n}+\etaThm_{S,n}(\lambda) =\gamevalueqcopt{\cG}+\epsilon(n)+\etaThm_{S,n}(\lambda),
        \end{aligned}
    \end{equation}
    where we define
    \begin{align}\label{eq:CompiledScoreNegligibleError}
        \etaThm_{S,n}(\lambda):=\|\vec{\beta}\|_1\left(\frac{C_G\eta_n(\lambda)}{1+\eta'_n(\lambda)}+\dim(V_n)\etaLem_n(\lambda)+\eta'_n(\lambda)\right).
    \end{align}
    Since $\dim(V_n)\leq\dim(\cB_n)$ is bounded independently of $\lambda$, $1+\eta'_n(\lambda)\geq1$, and $\eta_n$, $\etaLem_n$, and $\eta'_n$ are negligible, so is $\etaThm_{S,n}$.
    Its dependence on the QHE scheme and the QPT strategy $S$ is inherited from these functions.
\end{proof}

While \cref{thm:GeneralQuantBoundFiniteSecurityQCValue} provides upper bounds to the compiled score, it is fundamentally related to the NPA level $n$, which influences both the approximation error $\epsilon(n)$ and the negligible function $\etaThm_{S,n}(\lambda)$.
In general, a practically meaningful upper bound requires a sufficiently high NPA level $n$ so that $\epsilon(n)$ is small.
The computational cost of increasing $n$ is discussed in \cref{rem:WhyIsLogHere}.
The theorem is stated for each fixed $n$; choosing $n$ as a function of $\lambda$ would additionally require uniform control of the error estimates as $n$ grows.


Nonetheless, when the sequential NPA hierarchy converges at a finite level, the approximation error $\epsilon(n)$ vanishes at that level, and \cref{thm:GeneralQuantBoundFiniteSecurityQCValue} gives a quantitative soundness bound independent of the NPA approximation error.
A useful sufficient certificate for such finite convergence is the existence of a flat optimal solution.
Indeed, by \cref{thm:StoppingCriteriaSeqNPA}, a flat optimal solution at level $n_0$ implies $\gamevalueSeqNPA{\cG}{n_0}=\gamevalueqcopt{\cG}=\omega_{\mathrm q}(\cG)$ and gives a finite-dimensional optimal quantum strategy.
This leads to the following corollary.
\begin{corollary}\label{cor:MainBoundFiniteSecurityQCValue}
    Let $\cG$ be a bipartite Bell game.
    Consider its compiled version $\cG_{\comp}$ and let $S=(S_\lambda)_\lambda$ be an arbitrary quantum polynomial time (QPT) strategy employed by the prover.

    If the sequential NPA hierarchy for $\cG$ converges at some finite level $n_0$, i.e.,
    \begin{align}\label{eq:MainCorFiniteConvergence}
        \gamevalueSeqNPA{\cG}{n_0}=\gamevalueqcopt{\cG},
    \end{align}
    then there exists a negligible function $\etaThm_S(\lambda)$, depending on the QHE scheme and the strategy $S$, such that
    \begin{align}\label{eq:MainCorFirstEq}
        \gamevalueCompile{\lambda}{\cG_{\comp},S}\leq\omega_{\mathrm{qc}}(\cG)+\etaThm_S(\lambda),
    \end{align}
    where $\gamevalueCompile{\lambda}{\cG_{\comp},S}$ is the prover's Bell score using $S$ and $\omega_{\mathrm{qc}}(\cG)$ is the optimal commuting quantum score.

    In particular, if the sequential NPA hierarchy for $\cG$ admits a flat optimal solution (see \cref{def:FlatSequentialNPA}) at level $n_0$, then the hierarchy converges at level $n_0$ and $\cG$ admits an optimal finite-dimensional quantum strategy.
    Hence, $\gamevalueqcopt{\cG}=\omega_{\mathrm{q}}(\cG)$, where $\omega_{\mathrm{q}}(\cG)$ is the optimal tensor product quantum score, and
    \begin{align}\label{eq:MainCorSecondEq}
        \gamevalueCompile{\lambda}{\cG_{\comp},S}\leq\omega_{\mathrm{q}}(\cG)+\etaThm_S(\lambda).
    \end{align}
\end{corollary}
\begin{proof}
    Applying \cref{thm:GeneralQuantBoundFiniteSecurityQCValue} at level $n_0$ with $\epsilon(n_0)=0$ proves the first statement, with $\etaThm_S(\lambda):=\etaThm_{S,n_0}(\lambda)$.
    The flat optimality statement follows from \cref{thm:StoppingCriteriaSeqNPA}.
\end{proof}

The first part of \cref{cor:MainBoundFiniteSecurityQCValue} is a quantitative soundness statement with respect to the commuting quantum value.
The stronger flat-optimality assumption further gives quantitative quantum soundness with respect to the tensor product quantum value.

Moreover, in the flat-optimal case, \cref{thm:StoppingCriteriaSeqNPA} also extracts a finite-dimensional optimal quantum strategy.
The quantum completeness of the compiler from~\cite{kalai2023quantum} gives an honest compiled strategy attaining the same score up to negligible error, whenever the extracted finite-dimensional strategy satisfies the usual implementation assumptions of the compiler.
Together with the soundness bound for every fixed QPT strategy, this gives matching upper and lower bounds, with the negligible error allowed to depend on the strategy.

Although the finite-convergence and flat-optimality premises of \cref{cor:MainBoundFiniteSecurityQCValue} are game-specific (and determining if a game admits a finite-dimensional quantum realization is undecidable~\cite{fu2025membership}), they are intrinsic properties of the Bell game $\cG$ and its sequential NPA hierarchy, rather than assumptions on the compilation procedure.
Whenever such a finite-level certificate is available, the corollary applies directly.

Finally, we remark that an infinite-dimensional strategy is less well-posed in the computational setup: it is unclear how to implement such a strategy efficiently with $\poly(\lambda)$-size computers, and even if possible, a justification of the correctness of the QHE scheme in the infinite-dimensional setting is needed.

We end this subsection with two remarks, one on the more general Bell polynomials and one on the practical limit on the tightness of the bound in \cref{thm:GeneralQuantBoundFiniteSecurityQCValue}.
\begin{remark}\label{rem:MoreGeneralBellPolynomials}
    The derivations above focus on Bell polynomials $\vec{\beta}$ that are linear in the correlation $p^\lambda$ for simplicity.
    However, the same ideas extend to fixed polynomial objectives in the probabilities.
    The decomposition now applies to the modified correlation:
    \begin{align}\label{eq:ModifiedCorrelationDecomposition}
        \widetilde p^{\lambda,n}=\widetilde p^{\lambda,n}_{\mathrm{NS}}+\widetilde p^{\lambda,n}_{\mathrm{SI}}+\frac{\dim(V_n)\etaLem_n(\lambda)}{1+\dim(V_n)\etaLem_n(\lambda)}\widetilde p^{\lambda,n}_{\mathrm{res}}.
    \end{align}
    Both $p^\lambda$ and $\widetilde p^{\lambda,n}_{\mathrm{NS}}$ have entries in $[0,1]$.
    Thus, for every fixed integer $k\geq1$ and every $a,b,x,y$, we directly compute
    \begin{equation}\label{eq:HigherOrderCorrelationError}
        \begin{aligned}
            \left|(p^\lambda(ab|xy))^k-(\widetilde p^{\lambda,n}_{\mathrm{NS}}(ab|xy))^k\right|&\leq k\left|p^\lambda(ab|xy)-\widetilde p^{\lambda,n}_{\mathrm{NS}}(ab|xy)\right| \\
            & \leq k\left(\eta'_n(\lambda)+\frac{C_G\eta_n(\lambda)}{1+\eta'_n(\lambda)}+\frac{\dim(V_n)\etaLem_n(\lambda)}{1+\dim(V_n)\etaLem_n(\lambda)}\right),
        \end{aligned}
    \end{equation}
    The last inequality follows from \cref{eq:t1-correlation-error,eq:SICorrelationBound,eq:ResidualCorrelationFormula}.
    The same conclusion holds for every fixed polynomial in the probabilities, by expanding the differences of products.
    
    Similarly, the argument extends to fixed self-adjoint polynomials that are linear in Alice's measurements and contain Bob polynomials of degree at most $2n$, with terms of the form
    \begin{align}\label{eq:HigherDegreeBobScoreTerm}
        \sandwich{\Omega^\lambda_n}{\hA^{\lambda,n}\pi^\lambda_n(P(\B))}{\Omega^\lambda_n},\qquad P\in\cB_{2n}.
    \end{align}
    This follows from \cref{eq:DecompositionSignalingCoeffBound}, since the perturbation and residual terms are bounded by $2\eta'_n(\lambda)\coeffnorm{P}$ and $2\coeffnorm{P}$, respectively.
\end{remark}
\begin{remark}\label{rem:WhyIsLogHere}
    In the unreduced sequential NPA hierarchy, each moment matrix is indexed by words of length at most $n$ in Bob's measurement symbols and therefore has size $N=\sum_{j=0}^n(|I_B|\,|I_Y|)^j$.
    It is not necessarily the same as the dimension $\dim(V_n)$ in a particular GNS representation.
    For a fixed alphabet with $|I_B|\,|I_Y|>1$, $N$ grows exponentially in $n$.
    Thus even writing the unreduced matrices with size polynomial in $\lambda$ restricts the level to $n=O(\log\lambda)$.
    The cost of solving the SDP also depends on the input encoding and numerical accuracy; see~\cite{nesterov1994interior}.

    However, if the Bell game possesses sufficient symmetry or sparsity to reduce the effective matrix size, higher levels may become computationally accessible.
    This computational observation does not extend \cref{thm:GeneralQuantBoundFiniteSecurityQCValue} to a growing level $n=n(\lambda)$: that would additionally require uniform estimates for the negligible functions and their dependence on $n$.
\end{remark}

\subsection{Discussion on robust self-testing of compiled Bell games}\label{sec:RobustSelfTesting}
The authors of~\cite{kulpe2024bound} also present an exact self-testing of compiled Bell games.
We begin by introducing the notion of commuting operator self-testing following~\cite[Definition~7.1]{paddock2024operator}, and then recall the self-testing result.
\begin{definition}\label{def:CommutingOperatorSelfTesting}
    A nonlocal game $\cG$ with associated Bell polynomial $\beta$ is called a \emph{commuting operator self-test} if any commuting operator strategy that attains the optimal quantum commuting score, $\gamevalueqcopt{\cG}$, necessarily corresponds to the same ideal state $\rho^*$ on $\cA \otimes_{\max} \cB$.
\end{definition}
Note that this definition is a proper generalization of the standard self-testing when restricted to the states on the max tensor product of finite-dimensional $C^*$-algebras~\cite[Theorem~3.5]{paddock2024operator} up to the extremality condition.
But the infinite-dimensional case remains an open question.

Now we are ready to state the exact self-testing result for compiled Bell games~\cite[Theorem~6.5]{kulpe2024bound}.
\begin{theorem}\label{thm:KulpeSelfTesting}
    Let $\cG$ be a commuting operator self-test with the ideal state $\rho^*$.
    If $S$ is a QPT strategy for the compiled game $\cG_\comp$ such that $\lim_{\lambda \to \infty} \gamevalueCompile{\lambda}{\cG_\comp, S} = \gamevalueqcopt{\cG}$, then for the associated positive linear functional $\sigma^{\lambda}_{a|x}$ it holds that
    \begin{align}
        \lim_{\lambda \to \infty} \sigma^{\lambda}_{a|x}(P(\B)) = \rho^*(\A \otimes_{\max} P(\B) )
    \end{align}
    for every $x, a$ and every polynomial $P$.
    In particular,
    \begin{align}
        \lim_{\lambda \to \infty} \sigma^{\lambda}_{x}(P(\B)) = \rho^*(P(\B) ).
    \end{align}
\end{theorem}

Attempting to generalize all asymptotic results from~\cite{kulpe2024bound}, a natural question is whether we can generalize \cref{thm:KulpeSelfTesting} to the robust case with our quantitative framework.
However, as we discuss below, the current notions of robust self-testing have limitations that prevent us from establishing a robust generalization.
First, the following remark shows that a robust version of \cref{def:CommutingOperatorSelfTesting} is likely redundant.

\begin{remark}\label{rem:RobustSameExactCOSelfTest}
    In standard robust self-testing~\cite{zhao2024robust}, a necessary condition is that any finite-dimensional strategy $S$ achieving a Bell score within $\delta$ of the optimal quantum score $\omega_q^*(\cG)$ must have its associated state $\rho_S$ pointwise close to the ideal state $\rho^*$.
    Fixing the generators for $\cA$ and $\cB$ as Alice's and Bob's meausrement POVMs.
    
    Motivated by the standard robust self-testing, one might define a Bell game $\cG$ as a \emph{$\kappa$-robust commuting operator self-test} if there exists a function $\kappa:(0,\infty)\to[0,\infty)$, with $\kappa(\delta)\to0$ as $\delta\to0$, such that for every $\delta>0$ and every commuting operator strategy $S$ whose score satisfies $|\omega_S-\gamevalueqcopt{\cG}|\leq\delta$, its associated state $\rho_S$ satisfies
    \begin{align}\label{eq:RobustCoefficientDegreeBound}
        |\rho_S(P)-\rho^*(P)|\leq\deg(P)\coeffnorm{P}\kappa(\delta)
    \end{align}
    for every polynomial $P \in \cA\otimes_{\max}\cB$, with respect to some ideal state $\rho^*$.

    We now argue that this robust notion is redundant.
    On one hand, if the robust condition holds, the exact commuting operator self-testing property follows by letting $\delta\to0$ and using density of the polynomials.
    Conversely, suppose that the game $\cG$ is an exact self-test but not robust.
    By linearity and the triangle inequality, it suffices to prove \cref{eq:RobustCoefficientDegreeBound} for monomials.
    Non-robustness then implies that there exist a constant $c>0$, a sequence of states $\rho_n$ on $\cA\otimes_{\max}\cB$ with scores $\omega_n:=\rho_n(\beta)$ converging to $\gamevalueqcopt{\cG}$, and nonconstant monomials $w_n$ such that
    \begin{align}\label{eq:RobustFailureBoundedDegree}
        c\deg(w_n)\leq|\rho_n(w_n)-\rho^*(w_n)|\leq2\|w_n\|\leq2,
    \end{align}
    where $\coeffnorm{w_n}=1$ and the measurement POVM generators are contractions.
    Thus $\deg(w_n)\leq2/c$.
    Since there are only finitely many monomials of degree $\leq 2/c$, at least one of them occurs infinitely often in the sequence $(w_n)$.
    Passing to those indices and relabeling, we may therefore assume $w_n=w$ for a fixed nonconstant monomial $w$.
    
    By the Banach--Alaoglu theorem~\cite[\nopp II.6.2.7]{blackadar2006operator} and separability of $\cA\otimes_{\max}\cB$, there exists a weak-$^*$ convergent subsequence $\rho_{n_k}$ with limit $\rho$.
    Since $\rho(\beta)=\lim_{k\to\infty}\omega_{n_k}=\gamevalueqcopt{\cG}$, the exact self-testing property implies that $\rho$ coincides with the ideal state $\rho^*$.
    This contradicts the inequality $|\rho_{n_k}(w)-\rho^*(w)|\geq c\deg(w)>0$ for all $k$.
    Hence, the existence of a function $\kappa$ satisfying the robust condition in \cref{eq:RobustCoefficientDegreeBound} is equivalent to exact commuting operator self-testing in \cref{def:CommutingOperatorSelfTesting}.
\end{remark}

It is important to note that \cref{def:CommutingOperatorSelfTesting,thm:KulpeSelfTesting} applies within the framework of commuting quantum correlations (so does the standard finite-dimensional self-testing).
In our work, however, compiled Bell games $\cG_{\comp}$ at security parameter $\lambda$ are characterized using the sequential NPA hierarchy, which is a relaxation of the commuting quantum model.
Consequently, the current definitions of self-testing are too restrictive to fully capture the behavior of compiled Bell games.
This observation can serve as a motivation to develop a more general notion of robust self-testing capable of characterizing near-optimal scores even when the underlying correlations lie outside the strictly commuting set.
We note the potential connection to approximate Tsirelson's theorems~\cite{xu2025quantitative}, which characterize the distance of commuting to almost commuting correlations in finite dimensions.

\section{The sequential NPA hierarchy}\label{sec:SequentialNPA}
The sequential NPA hierarchy, which we now formally introduce, is the central analytical tool underpinning our quantitative soundness bounds from \cref{sec:QuantitaiveBoundConvergentRateMain}.
It provides a natural adaptation of the standard NPA framework to the setting of sequential Bell games, as depicted in \cref{fig:NonlocalCompiledBellGame}.(b), and steering scenarios.
This hierarchy models a scenario where provers are queried sequentially under a strong no-signaling condition, which prevents the second prover's actions from depending on the first prover's question.

In this formulation, for each $a, x$ we define a subnormalized moment matrix $\Theta^{(n)}(a|x)$ for monomials in the letters $\{ \B \}$ with length $\leq n$, and consider the normalized moment matrix $\Theta^{(n)} = \sum_{a} \Theta^{(n)}(a|x)$.
The corresponding SDP relaxation is given by
\begin{equation}\label{eq:SequentialNPASDP}
    \begin{aligned}
        \gamevalueSeqNPA{\cG}{n} \quad &= \max_{\Theta^{(n)}(a|x) \geq 0 \, \forall a,x} \langle \vec{\beta},\, p \rangle \\
        \text{subject to} \quad
        & p(ab|xy) = \Theta^{(n)}(a|x)_{1, \B} \quad \forall a, b, x, y \quad \text{(probability extraction)}, \\
        & 0 \leq \B \leq \id \quad \forall b, y \quad \text{(via localizing matrices; POVM bounds for Bob)}, \\
        & \sum_b \B = \id \quad \forall y \quad \text{(via localizing matrices; POVM completeness for Bob)}, \\
        & \sum_{a} \Theta^{(n)}(a|x) = \sum_{a} \Theta^{(n)}(a|x') := \Theta^{(n)} \quad \forall x, x' \quad \text{(strongly no-signaling condition)}, \\
        & 1 = \Theta^{(n)}_{\id, \id} \quad \text{(normalization)}.
    \end{aligned}
\end{equation}
The localizing matrices are those indexed by monomials in $\cB_{n-1}$.
For every $n$, this SDP directly corresponds to the compiled Bell game in the asymptotic security limit (i.e., $\lambda \to \infty$), via the identification
\begin{align}
    \sigma^{\lambda \to \infty, n}_{a|x}(w^*v) = \Theta^{(n)}(a|x)_{w, v}.
\end{align}
This is a convergent SDP hierarchy to the optimal commuting quantum score $\gamevalueqcopt{\cG}$ from above, as formalized in \cref{thm:SequentialNPAConvergence} (where the convergence of a modified NPA hierarchy is also shown).

Note that we use POVMs for Bob, following~\cite{kulpe2024bound} and the construction in~\cref{sec:QuantitaiveBoundConvergentRateMain}.
A projective variant can be defined by additionally imposing $B_{b|y}B_{b'|y}=\delta_{b,b'}B_{b|y}$ for all $b,b'\in I_B$ and $y\in I_Y$.
This can potentially give tighter bounds at finite levels than the POVM version.

Having defined the hierarchy, we dedicate the remainder of this section to its full characterization.
We compare it with the standard NPA hierarchy (\cref{prop:CompareStandardSequentialNPA,thm:SequentialNPAConvergence}), establish its stopping criterion (\cref{thm:StoppingCriteriaSeqNPA}) and showcase a concrete example of flat optimal solution (\cref{ex:CHSHFlatSequentialNPA}), and identify its conic dual as a special case of the sparse SOS hierarchy~\cite{klep2022sparse} (\cref{prop:SequentialDualSparse}).

\subsection{Comparison with the standard NPA hierarchy}\label{sec:SequentialVsStandardNPA}
It is natural to compare the sequential NPA hierarchy defined in \cref{eq:SequentialNPASDP} to the \emph{standard NPA hierarchy}, which we recall now.
Here, the moment matrix $\Gamma^{(n)}$ is constructed from monomials in the letters $\{\A, \B\}$ of length $\leq n$.
The associated SDP reads as follows:
\begin{equation}\label{eq:StandardNPASDP}
    \begin{aligned}
        \gamevalueNPA{\cG}{n} \quad &= \max_{\Gamma^{(n)} \geq 0} \langle \vec{\beta},\, p \rangle \\
        \text{subject to} \quad
        & p(ab|xy) = \Gamma^{(n)}_{\A, \B} \quad \forall a, b, x, y \quad \text{(probability extraction)}, \\
        & 0 \leq \A, \B \leq \id \quad \forall a, b, x, y \quad \text{(POVM bounds)}, \\
        & \sum_a \A = \sum_{b} \B = \id \quad \forall x, y \quad \text{(POVM completeness)}, \\
        & A_{a|x}A_{a'|x}=\delta_{a,a'}A_{a|x}, \quad B_{b|y}B_{b'|y}=\delta_{b,b'}B_{b|y} \quad \text{(projecitivity)}, \\
        & [\A, \B] = 0 \quad \forall a, b, x, y \quad \text{(commutation)}, \\
         & 1 = \Gamma^{(n)}_{\id, \id} \quad \text{(normalization)}.
    \end{aligned}
\end{equation}
Note this is stated with PVMs rather than POVMs in the sequential NPA hierarchy.

A feasible solution to the standard NPA hierarchy at level $n$ can be mapped to a feasible solution for the sequential NPA hierarchy at level $n-1$ by setting
\begin{align}
    \Theta^{(n-1)}(a|x)_{w, v} = \Gamma^{(n)}_{w, \A \cdot v}
\end{align}
for all $w, v$ monomials in $\cB_{n-1}$.
Therefore, an approximation-error bound for the sequential NPA hierarchy at level $n-1$ also gives an approximation-error bound for the standard NPA hierarchy at level $n$.
We next compare the sequential NPA hierarchy with a modified NPA hierarchy.
\begin{proposition}\label{prop:CompareStandardSequentialNPA}
    Consider the modified NPA hierarchy yielding a score $\gamevalueModNPA{\cG}{n}$ defined by
    \begin{equation}\label{eq:ModifiedNPASDP}
        \begin{aligned}
             \gamevalueModNPA{\cG}{n} \quad &= \max_{\tilde{\Gamma}^{(n)} \geq 0}\langle \vec{\beta},\, p \rangle \\
            \text{subject to} \quad
            & p(ab|xy) = \tilde{\Gamma}^{(n)}_{\A, \B} \quad \forall a, b, x, y \quad \text{(probability extraction)}, \\
            & 0 \leq \A, \quad 0 \leq \B \leq \id \quad \forall a, b, x, y \quad \text{(POVM bounds)}, \\
            & \A\B \geq 0,\qquad \A(\id-\B)\geq0 \quad\text{(via localizing matrices)}, \\
            & \sum_{b} \B = \id \quad \forall x, y \quad \text{(POVM completeness for Bob)}, \\
            & \sum_a \tilde{\Gamma}^{(n)}_{b_1, \A b_2} = \tilde{\Gamma}^{(n)}_{b_1, b_2} \quad \forall b_1 \in \cB_n, b_2 \in \cB_{n-1} \quad \text{(Alice fakes POVM properties to Bob)}, \\
            & [\A, \B] = 0 \quad \forall a, b, x, y \quad \text{(commutation)}, \\
            & 1 = \tilde{\Gamma}^{(n)}_{\id, \id} \quad \text{(normalization)}.
        \end{aligned}
    \end{equation}
    The localizing matrices for $\A \B$ and $\A (\id - \B)$ are indexed by monomials in $\B$ of degree at most $n-2$ (omitted when $n=1$).
    Compared with the standard NPA hierarchy, we drop the projectivity conditions, keep only positivity for Alice's operators, and impose Alice's completeness only against Bob's polynomials as specified above.
    Note that \cref{eq:ModifiedNPASDP} is a relaxation of the standard NPA hierarchy in \cref{eq:StandardNPASDP} at level $n$ with
    \begin{align}
        \gamevalueNPA{\cG}{n} \leq \gamevalueModNPA{\cG}{n}.
    \end{align}

    For every $n\geq2$, the existence of a modified NPA moment matrix $\tilde{\Gamma}^{(n)}$ implies the existence of a strongly no-signaling sequential NPA moment matrix $\Theta^{(n-1)}$ with the same correlation $p$.
    Conversely, the existence of $\Theta^{(n)}$ implies the existence of $\tilde{\Gamma}^{(n-1)}$ with the same correlation $p$.
    It follows that, for all $n \geq 2$,
    \begin{align}
       \gamevalueSeqNPA{\cG}{n+1} \leq \gamevalueModNPA{\cG}{n} \leq \gamevalueSeqNPA{\cG}{n-1}.
    \end{align}
    Consequently, the modified NPA hierarchy also asymptotically converges to $\gamevalueqcopt{\cG}$.
\end{proposition}
\begin{proof}
    Given a modified moment matrix $\tilde{\Gamma}^{(n)}$, we claim that we get a feasible solution for the sequential NPA hierarchy at level $n-1$ by
    \begin{align}\label{eq:ModifiedToSequentialMoments}
        \Theta^{(n-1)}(a|x)_{w,v}:=\tilde{\Gamma}^{(n)}_{w,\A v},\qquad w,v\in\cB_{n-1}.
    \end{align}
    Positivity of each $\Theta^{(n-1)}(a|x)$ follows from the localizing matrix for $\A\geq0$.
    Its Bob localizing matrices are
    \begin{align}\label{eq:ModifiedConditionalLocalizingMatrices}
        \left(\tilde{\Gamma}^{(n)}_{u,\A\B v}\right)_{u,v\in\cB_{n-2}}\succeq0,\qquad
        \left(\tilde{\Gamma}^{(n)}_{u,\A v}-\tilde{\Gamma}^{(n)}_{u,\A\B v}\right)_{u,v\in\cB_{n-2}}\succeq0,
    \end{align}
    which are exactly the positive semidefinite localizing matrices for $\A\B$ and $\A(\id-\B)$.
    Note that the weak completeness is already sufficient to ``fake'' the strongly no-signaling condition, and the completeness and normalization for $\B$ inherits directly from $\tilde{\Gamma}^{(n)}$.
    
    For the converse direction, suppose we have $\Theta^{(n)}$, one may identify this with a compiled Bell game with strongly no-signaling condition via
    \begin{equation}\label{eq:StateMomentIdentification}
        \begin{aligned}
            &\sigma^{n}_{a|x}(w^*v) := \Theta^{(n)}(a|x)_{w, v} \quad \forall a, x\\
            &\sigma^{n}_x(w^*v) = \sigma^n(w^*v) := \Theta^{(n)}_{w, v} \quad \forall x
        \end{aligned}
    \end{equation}
    as positive linear maps $\cB_{2n} \to \mathds{C}$.
    We then use the same flat extension technique as in \cref{sec:FlatExtension} and~\ref{sec:FiniteSecurityGNSStrategy} to obtain positive functionals $\sigma_{a|x}: \cB \to \mathds{C}$ with $\sigma_x = \sum_a\sigma_{a|x}$.
    As extensions, the linear functionals $\sigma_{a|x}$ agree with $\sigma^{n}_{a|x}$ on the subspace $\cB_{2n-2}$, so the states $\sigma_x$ agree with $\sigma^n$ on $\cB_{2n-2}$.
    A crucial observation is that $\sigma_x \neq \sigma_{x'}$ in general, in contrast to their behaviors in $\cB_{2n-2}$.

    Then, using \cref{prop:GNSConstructionCommonSpace} for $\{\sigma_x\}_x$ at level $n-1$ with $\eta_{n-1}=0$ we have:
    \begin{enumerate}
        \item A GNS representation $(H_{n-1}, \pi_{n-1}, \ket{\Omega_{n-1}})$.
        \item The operators $\{\pi(\B)\}$ form POVMs in $B(\cH_{n-1})$.
        \item Positive operators $\hA^{n-1} \in \pi_{n-1}(\cB)' \subset B(\cH_{n-1})$ for all $a, x$ such that
        \begin{align}
            \Theta^{(n)}(a|x)_{w,v} = \sandwich{\Omega_{n-1}}{\hA^{n-1} \pi_{n-1}(w^* v)}{\Omega_{n-1}}
        \end{align}
        for $w, v \in \cB_{n-1}$.
        Note that, however, the equation need not hold when $w, v \in \cB_{n} \setminus \cB_{n-1}$ because the flat extension technique affects these entries.
        \item The operators $\hA^{(n-1)}$ behave like POVMs for low-degree polynomials of Bob's measurements, i.e., for any $P_1, P_2 \in \cB_{n-1}$, one has
        \begin{align}
             \sandwich{\Omega_{n-1}}{\pi^{n-1}(P_1) (\sum_a\hA^{(n-1)} - \id_{H^{n-1}}) \pi^{n-1}(P_2)} {\Omega_{n-1}} = 0.
        \end{align}
        But the above equation does not hold for $P_1, P_2$ of higher degrees, due to the extensions $\sigma_x \neq \sigma_{x'}$ for higher degree polynomials.
    \end{enumerate}
    One can then identify the letter $\A$ with $\hA^{(n-1)}$ and $\B$ with $\pi^{n-1}(\B)$, and check that the formula
    \begin{align}\label{eq:ReconstructedModifiedMomentMatrix}
        \tilde{\Gamma}^{(n-1)}_{w, v} = \sandwich{\Omega_{n-1}}{w^* v} {\Omega_{n-1}}
    \end{align}
    defines a modified moment matrix $\tilde{\Gamma}^{n-1}$.
\end{proof}

Finally, the sequential NPA hierarchy is asymptotically equivalent to the standard NPA hierarchy: both converge to the optimal commuting quantum score $\gamevalueqcopt{\cG}$ from above.
\begin{theorem}\label{thm:SequentialNPAConvergence}
    The following statements are equivalent:
    \begin{enumerate}[label=(\roman*)]
        \item The correlations $p(ab|xy)$ arise from a commuting-observable strategy.
        \item There exists a family of $\{{\Gamma}^{(n)}\}_n$ of feasible solutions to the standard NPA hierarchy \cref{eq:StandardNPASDP} such that $p(ab|xy) = {\Gamma}^{(n)}_{\A, \B}$ for all $n$.
        \item There exists a family of $\{\Theta^{(n)}(a|x), \forall a, x\}_n$ of feasible solutions to the sequential NPA hierarchy \cref{eq:SequentialNPASDP} such that $p(ab|xy) = \Theta^{(n)}(a|x)_{\id, \B}$ for all $n$.
        \item There exists a family of $\{\tilde{\Gamma}^{(n)}\}_n$ of feasible solutions to the modified NPA hierarchy \cref{eq:ModifiedNPASDP} such that $p(ab|xy) = \tilde{\Gamma}^{(n)}_{\A, \B}$ for all $n$.
    \end{enumerate}
    Consequently, both $\gamevalueSeqNPA{\cG}{n} \searrow \gamevalueqcopt{\cG}$ and $\gamevalueModNPA{\cG}{n} \searrow \gamevalueqcopt{\cG}$ as $n \to \infty$.
\end{theorem}
\begin{proof}
    \textit{(i)} $\iff$ \textit{(ii)} is due to~\cite{navascues2008convergent}.
    \textit{(ii)} $\implies$ \textit{(iii)} is trivial, while the converse is an immediate consequence of the proof for \cref{prop:CompareStandardSequentialNPA}: that as $n \to \infty$, the modified NPA hierarchy \cref{eq:ModifiedNPASDP} coincides with the standard NPA hierarchy \cref{eq:StandardNPASDP}.
    \textit{(iii)} $\iff$ \textit{(iv)} is due to \cref{prop:CompareStandardSequentialNPA}.
\end{proof}

\subsection{Stopping criterion for the sequential NPA hierarchy}\label{sec:StoppingCriteriaSeqNPA}
We now discuss the stopping criterion for the sequential NPA hierarchy.
First introduced in \cref{eq:FlatConditionSec2}, we define more precisely the flatness condition for the sequential NPA hierarchy and then show its consequences in relation to the finite-dimensional quantum realizations.
In \cref{ex:CHSHFlatSequentialNPA}, we give an example based on CHSH game of when the sequential NPA hierarchy admits a flat optimal solution.

\begin{definition}\label{def:FlatSequentialNPA}
    Let $\{\Theta^{(n)}(a|x)\}$ be the solution of the sequential NPA hierarchy at level $n$ from \cref{eq:SequentialNPASDP} for a Bell game $\cG$.
    Denote $\Theta^{(n)} = \sum_a \Theta^{(n)}(a|x)$ and consider its block form
    \begin{align}
        \Theta^{(n)} = \begin{pmatrix}
            \Theta^{(n-1)} & B \\
            B^* & C
        \end{pmatrix},
    \end{align}
    where $\Theta^{(n-1)}$ is the block indexed by monomials of degree $\leq n-1$, and $C$ is the block indexed by monomials of degree exactly $n$.
    Then we say the solution $\{\Theta^{(n)}(a|x)\}$ is flat (or has a rank-loop) if
    \begin{align}
        \mathrm{rank}(\Theta^{(n)}) = \mathrm{rank}(\Theta^{(n-1)}) < \infty.
    \end{align}
\end{definition}

This leads to our second main theorem.
\begin{theorem}\label{thm:StoppingCriteriaSeqNPA}
    Let $\cG$ be a bipartite Bell game, and suppose that the sequential NPA hierarchy for $\cG$ has a flat optimal solution $\{\Theta^{(n)}(a|x)\}$ at some finite level $n$.
    Then this solution yields a finite-dimensional GNS representation $(\cH, \pi, \ket{\Omega})$ of $\cB$ and an optimal commuting quantum strategy $(\hA, \pi(\B), \ket{\Omega})$ satisfying:
    \begin{enumerate}[label=(\roman*)]
        \item There exist POVMs $\{\hA\}_{a,x} \subset \pi(\cB)' \subset B(\cH)$, where $\pi(\cB)'$ is the commutant of $\pi(\cB)$.
        \item Bob's measurements in this representation, $\{\pi(\B)\}_{b,y}$, are POVMs.
        \item The probability distribution $p(ab|xy) = \Theta^{(n)}(a|x)_{1, \B}$ from \cref{eq:SequentialNPASDP} is recovered by Born's rule in this representation, i.e.,
        \begin{align}
            p(ab|xy) = \sandwich{\Omega}{\hA \pi(\B)}{\Omega}.
        \end{align}
    \end{enumerate}
    In particular, this strategy is equivalent to a finite-dimensional tensor product strategy, and
    \begin{align}
        \gamevalueSeqNPA{\cG}{n} = \gamevalueqcopt{\cG} = \omega_{\mathrm{q}}(\cG).
    \end{align}
\end{theorem}
\begin{proof}
    Assume that $\{\Theta^{(n)}(a|x)\}$ is a flat optimal solution at level $n$.
    Using \cref{eq:StateMomentIdentification} we identify the moment matrix $\Theta^{(n)}$ with a positive linear functional $\sigma^n: \cB_{2n} \to \mathds{C}$ and each $\Theta^{(n)}(a|x)$ with a $\sigma_{a|x}^n: \cB_{2n} \to \mathds{C}$.
    
    First, we show that every $\Theta^{(n)}(a|x)$ is also flat.
    Since $\Theta^{(n)}$ is flat, its corresponding functional $\sigma^n$ can be extended to a state $\sigma$ on $\cB$ via a finite-dimensional GNS representation $(\cH, \pi, \ket{\Omega})$ (\cref{prop:GNSfromFlatExtension}).
    The flatness condition means:
    \begin{align}
        \cH = \mathrm{span}\{ \pi(P) \ket{\Omega} \mid P \in \cB_{n}\} = \mathrm{span}\{ \pi(P) \ket{\Omega} \mid P \in \cB_{n-1}\}.
    \end{align}
    This equality implies that for every monomial $w \in \cB_n \setminus \cB_{n-1}$, we have a linear dependence
    \begin{align}
        \pi(w)\ket{\Omega} = \sum_{v \in \cB_{n-1}} c_v \pi(v) \ket{\Omega}
    \end{align}
    for some constants $c_v \in \mathds{C}$.
    It follows that for the polynomial $P_w = w - \sum_{v \in \cB_{n-1}} c_v v \in \cB_n$,
    \begin{align}
        0 = \lVert \pi(P_w) \ket{\Omega} \rVert^2 = \sigma(P_w^* P_w) = \sigma^n(P_w^* P_w) = \sum_a \sigma_{a|x}^n(P_w^* P_w)
    \end{align}
    for all $x$.
    Hence $\sigma_{a|x}^n(P_w^* P_w) = 0$ for all $a, x$ by positivity.
    Moreover, the Cauchy-Schwarz inequality implies that
    \begin{align}
        \sigma(P_w) = \sigma^n(P_w) = \sigma_{a|x}^n(P_w) = \sigma_{a|x}(P_w) = 0.
    \end{align}
    But the condition $\sigma_{a|x}^n(P_w^* P_w) = 0$ for the same polynomials $P_w$ means that the Gram vectors corresponding to monomials in $\cB_n \setminus \cB_{n-1}$ for each $\Theta^{(n)}(a|x)$ satisfy the same linear dependence relations on Gram vectors from $\cB_{n-1}$.
    That is, all $\Theta^{(n)}(a|x)$ are flat in the same block form.

    Next, we construct Alice's operators.
    Denote by $\sigma_{a|x}: \cB \to \mathds{C}$ the flat extension of $\sigma_{a|x}^n$ in the sense of \cref{sec:FlatExtension}.
    Following standard arguments (\cref{sec:FiniteSecurityGNSStrategy} and \cref{prop:CompareStandardSequentialNPA}), we can construct positive operators $\hA \in \pi(\cB)' \subset B(\cH)$ such that for any $Q \in \cB$ and $w,v \in \cB_{n}$:
    \begin{align}
        \sigma_{a|x}(Q) = \sandwich{\Omega}{\hA \pi(Q)}{\Omega} \text{ and } \Theta^{(n)}(a|x)_{w,v} = \sandwich{\Omega}{\hA \pi(w^* v)}{\Omega}.
    \end{align}
    (This equality holds for $w,v \in \mathcal{B}_n$, as opposed to $\cB_{n-1}$ in the proof of \cref{prop:CompareStandardSequentialNPA}, precisely because $\Theta^{(n)}(a|x)$ have been shown to be flat.)
    Statements \textit{(ii)} and~\textit{(iii)} then straightforwardly follow.

    To show that $\{\hA\}$ are actually POVMs for each $x$, it suffices to show that the state $\sigma_x := \sum_a \sigma_{a|x}$ is equal to $\sigma$ for all $x$.
    (We refer to the proofs of \cref{prop:GNSConstructionCommonSpace,prop:CompareStandardSequentialNPA} for this equivalence.)
    To this end, it is useful to recall what flat extension from $\sigma^n$ to $\sigma$ does exactly: consider the set of null polynomials $P_w = w - \sum_{v \in \cB_{n-1}} c_v v$ for $w \in \cB_n \setminus \cB_{n-1}$, generating a two-sided ideal $J$ for which $\sigma(J)=0$.
    Then, for any $Q \in \cB$, there exists a low-degree representative $Q' \in \cB_{2n-2}$ such that $Q - Q' \in J$.
    The flat extension is then constructed via the equation $\sigma(Q) = \sigma(Q') = \sigma^n(Q')$.
    (For example, $Q = w, Q' = \sum_{v \in \cB_{n-1}} c_v v$ with $P_w = Q - Q'$.)

    On the other hand, each $\sigma_{a|x}$ is extended from $\sigma_{a|x}^n$ using another two-sided ideal $J_{a|x}$.
    We have already shown that $\sigma_{a|x}^n(P_w^* P_w) = 0$, hence all $P_w \in J_{a|x}$ and $J \subset J_{a|x}$ for all $a, x$.
    Consequently, if $Q - Q' \in J$, then $Q - Q' \in \bigcap_a J_{a|x}$, which implies that
    \begin{align}
        \sigma_x(Q) = \sum_a \sigma_{a|x}(Q) = \sum_a \sigma_{a|x}(Q') = \sum_a \sigma_{a|x}^n(Q') = \sigma^n(Q') = \sigma(Q') = \sigma(Q).
    \end{align}
    It follows that $\sigma_x = \sigma$ for all $x$ since $Q \in \cB$ was arbitrary.

    We have now shown that $(\hA, \pi(\B), \ket{\Omega})$ is a finite-dimensional quantum strategy with commuting observables achieving the Bell score $\gamevalueSeqNPA{\cG}{n}$.
    By definition of the sequential NPA hierarchy as a relaxation, $\gamevalueSeqNPA{\cG}{n} \geq \gamevalueqcopt{\cG}$.
    Conversely, $\gamevalueqcopt{\cG}$ is the optimal value over quantum commuting observable strategies, so $\gamevalueSeqNPA{\cG}{n} \leq \gamevalueqcopt{\cG}$.
    This, along with Tsirelson's theorem for finite-dimensional commuting strategies (see, e.g.,~\cite{scholz2008tsirelson, xu2025quantitative}), proves the ``in particular'' statement.
\end{proof}

Thus a flat optimal solution is a finite-level certificate that the hierarchy has stopped and that the optimum is realized by a finite-dimensional quantum strategy.
We do not claim the converse.
A finite-dimensional optimal strategy does produce a sequence of feasible solutions to the sequential hierarchy that is flat~\cite[Theorem~10]{navascues2008convergent}, but this alone does not exclude the possibility of having a feasible non-flat solution at the same finite level.
Moreover, finite convergence of the values need not imply flat optimality: for example, a game with value $1$ but no finite-dimensional perfect strategy has trivial value convergence but cannot have a flat optimal solution.

In addition, regarding the numerical implementation, having an optimal finite-dimensional strategy does not guarantee that the sequential NPA hierarchy will find a flat optimal solution in practice.
In fact, it is possible that there exist infinitely many inequivalent finite-dimensional optimizers, leading the SDP solver for the sequential NPA hierarchy to freely return any convex mixtures of them.
We further remark that the decision problem of whether a correlation admits a finite-dimensional quantum realization is undecidable in general~\cite{fu2025membership}.

\begin{remark}
 A feature of the sequential NPA hierarchy \cref{eq:SequentialNPASDP} is that all constraints are of degree one, thus it suffices to check flatness over the block $\Theta^{(n-1)}$.
If adding higher order polynomial constraints $Q(\{\B\})$ of $\deg(Q)=d$ to \cref{eq:SequentialNPASDP}, the result of \cref{thm:StoppingCriteriaSeqNPA} will remain valid if we change the flatness condition to $\mathrm{rank}(\Theta^{(n)}) = \mathrm{rank}(\Theta^{(n-d)})$, where $\Theta^{(n-d)}$ is the block indexed by monomials of degree $\leq n-d$.
\end{remark}

We next give an explicit example in which the sequential NPA hierarchy with POVMs admits a flat optimal solution.

\begin{example}[The CHSH game]\label{ex:CHSHFlatSequentialNPA}
Consider inputs $x,y\in\{0,1\}$ and outputs $a,b\in\{-1,1\}$.
We show that the sequential NPA relaxation for the CHSH expression is exact at level $2$ and admits a flat optimal solution at level $3$, with only POVM constraints.

Let $\sigma_{a|x}:\cB_4\to\mathds{C}$ be the truncated functionals corresponding to a feasible level-$2$ solution, and let $\sigma:=\sum_a\sigma_{a|x}$ be their common normalized marginal.
Define Bob's observables by $B_y:=B_{1|y}-B_{-1|y}$.
The CHSH expression is
\begin{align}\label{eq:CHSHSequentialExpression}
    C:=\sum_{a=\pm1}a\bigl[\sigma_{a|0}(B_0+B_1)+\sigma_{a|1}(B_0-B_1)\bigr].
\end{align}

We first derive $\sigma(B_y^2)\leq1$ from the POVM localizing constraints.
Since $B_{1|y}=(\id+B_y)/2$ and $B_{-1|y}=(\id-B_y)/2$, we have
\begin{equation}\label{eq:CHSHSequentialLocalizingIdentity}
\begin{aligned}
    &(\id-B_y)B_{1|y}(\id-B_y)+(\id+B_y)B_{-1|y}(\id+B_y)\\
    &\quad=\frac12(\id-B_y)^2(\id+B_y)+\frac12(\id+B_y)^2(\id-B_y) = \id-B_y^2.
\end{aligned}
\end{equation}
Each term in the first line has the form $Q^*B_{b|y}Q$ with $Q$ of degree one.
Thus the level-$2$ localizing constraints for positivity of $\B$ imply $0\leq\sigma(\id-B_y^2)=1-\sigma(B_y^2)$.
With the positivity of the $\sigma_{a|x}$, one directly computes
\begin{equation}\label{eq:CHSHSequentialUpperBound}
\begin{aligned}
    0 &\leq\sum_{x=0}^1\sum_{a=\pm1}
    \sigma_{a|x}\left(\left(a\id-\frac{B_0+(-1)^xB_1}{\sqrt{2}}\right)^2\right)\\
    &=\sum_{x,a}\sigma_{a|x}(\id)-\sqrt{2}C
      +\frac12\sigma\bigl((B_0+B_1)^2+(B_0-B_1)^2\bigr)\\
    &=2-\sqrt{2}C+\sigma(B_0^2+B_1^2) \leq4-\sqrt{2}C.
\end{aligned}
\end{equation}
Thus $C\leq2\sqrt{2}$.

To construct an optimal solution, we use the usual two-qubit CHSH strategy.
Let $X$ and $Z$ be the Pauli matrices on $\mathds{C}^2$, and take
\begin{equation}\label{eq:CHSHSequentialStrategy}
\begin{aligned}
    \ket{\Phi^+}&:=\frac{\ket{00}+\ket{11}}{\sqrt{2}},\\
    A_0&:=\frac{Z+X}{\sqrt{2}}, \qquad A_1:=\frac{Z-X}{\sqrt{2}}, \qquad B_0:=Z, \qquad B_1:=X.
\end{aligned}
\end{equation}
The corresponding measurement effects are
\begin{align}\label{eq:CHSHSequentialEffects}
    A_{a|x}:=\frac{\id+aA_x}{2},
    \qquad B_{b|y}:=\frac{\id+bB_y}{2}.
\end{align}
For each level $n$, define
\begin{align}\label{eq:CHSHSequentialWitness}
    \Theta^{(n)}(a|x)_{w,v}:=\sandwich{\Phi^+}{A_{a|x}\otimes w^*v}{\Phi^+},
\end{align}
where $w,v$ are words of length at most $n$ in Bob's POVM effects, evaluated at the matrices in \cref{eq:CHSHSequentialEffects}.
These moment matrices give rise to feasible solutions of the sequential NPA hierarchy at every level.

Its CHSH value is
\begin{equation}\label{eq:CHSHSequentialAttainment}
\begin{aligned}
    C &=\sandwich{\Phi^+}{A_0\otimes(B_0+B_1)+A_1\otimes(B_0-B_1)}{\Phi^+}\\
    &=\sqrt{2}\sandwich{\Phi^+}{Z\otimes Z+X\otimes X}{\Phi^+} =2\sqrt{2}.
\end{aligned}
\end{equation}
Together with \cref{eq:CHSHSequentialUpperBound}, this proves exactness at level $2$ and optimality of these moment matrices at every level $n\geq2$.

Finally, we verify flatness.
Since Bob's marginal state is $\id/2$, the matrix $\Theta^{(n)}=\sum_a\Theta^{(n)}(a|x)$ satisfies
\begin{align}\label{eq:CHSHSequentialGram}
    \Theta^{(n)}_{w,v} =\sandwich{\Phi^+}{\id\otimes w^*v}{\Phi^+} =\frac{1}{2}\Tr{w^*v}.
\end{align}
Thus it is the Gram matrix for Bob's words under the normalized Hilbert--Schmidt inner product, and so its rank equals the dimension of their linear span.
The words of length at most $1$ span $\operatorname{span}\{\id,X,Z\}$, while the words of length at most $2$ already span all of $\operatorname{span}\{\id,X,Z,ZX\}=M_2(\mathds{C})$.
Consequently,
\begin{align}\label{eq:CHSHSequentialFlatness}
    \operatorname{rank}\Theta^{(1)}=3, \qquad \operatorname{rank}\Theta^{(2)} =\operatorname{rank}\Theta^{(3)}=4.
\end{align}
Thus the level-$3$ solution given by \cref{eq:CHSHSequentialWitness} is flat and optimal.
\end{example}
The same argument extends to all two-player XOR games.
The sum-of-squares certificates in~\cite[Theorem~3.8]{cui2024computational}, together with the POVM localizing estimate in~\cref{eq:CHSHSequentialLocalizingIdentity}, show that the sequential NPA hierarchy is exact at level~$2$.
Moreover, every such game admits a finite-dimensional optimal quantum strategy~\cite[Corollary~2.11]{cui2024computational}.
Hence the associated optimal moment matrices have bounded nondecreasing ranks and therefore give a flat optimal solution at some finite level.

\subsection{Sequential NPA hierarchy is conic dual to sparse SOS hierarchy}\label{sec:SparseDualSequential}
Another natural question is to ask what the dual of the sequential NPA hierarchy is, i.e., what is the corresponding sum of squares (SOS) certificate.
Its conic dual has the sparse SOS structure of~\cite{klep2022sparse}, with each square involving one Alice input and Bob's generators.
This conic duality correspondence provides further characterization of the sequential NPA hierarchy and insights into its numerical performance from the sparse SOS numerical examples~\cite[Chapter~6.7]{magron2023sparse}.

In order to formulate the conic dual of the sequential NPA hierarchy at level $n$, we first restrict our attention to the polynomial space generated by the measurement operators $\A, \B$.
Specifically, note that the sequential NPA hierarchy at level $n$ characterizes polynomials that are at most of degree $2n$ in $\B$ and only linear in $\A$ (via the matrices $\Theta^{(n)}(a|x)$).
Thus, the natural polynomial vector space is
\begin{align}
    V_{(n)} = \{ \sum_{a,x} \A f_{a|x}(\B) + g(\B) \mid f_{a|x}, g \in \cB_{2n} \}.
\end{align}
For the duality proof, Alice's letter are self-adjoint symbols $\A$ satisfying $\A A_{a'|x} = \delta_{a,a'} \A$ and commuting with Bob's letter $\B$.
Alice's POVM completeness is imposed below, while Bob's $\B$ are not required to be projective.
In this polynomial space $V_{(n)}$, the sparse SOS cone at level $n$ is then defined as
\begin{align}\label{eq:SparseSequentialCone}
    \mathcal{M}_{(n)} = \biggl\{ \sum_{x,i} \Bigl( \sum_{a} \A f_{a|x,i} + g_{x, i} \Bigr)^* t_{x,i} \Bigl( \sum_{a} \A f_{a|x,i} + g_{x,i} \Bigr) + \sum_x \Bigl(1 - \sum_{a} \A \Bigr) h_x\biggr\},
\end{align}
where each $t_{x,i}$ is chosen from $\{\id, \B, \id - \B \mid b, y\}$, each sum over $i$ is finite, $h_x = h^* \in \cB_{2n}$, and $f_{a|x,i}, g_{x,i} \in \cB_n$ if $t+{x,i} = \id$ and in $\cB_{n-1}$ otherwise.
The terms with $t_{x,i}$ correspond to the moment matrices, while those with $t_{x,i} = \B$ or $\id - \B$ correspond to Bob's localizing matrices.
(Note that no products of Alice's $\A$ with different x appears in this cone.)

If the Bell polynomial $\beta$ can be identified with an element in $V_{(n)}$, then the sparse SOS hierarchy at level $n$, yielding a score $\gamevalueSparse{\cG}{n}$, is given by:
\begin{equation}\label{eq:sparseSOSSDP}
    \begin{aligned}
        \gamevalueSparse{\cG}{n} := \inf_{m \in \mathds{R},s\}} \quad & \, m \\
        \text{subject to } \quad m\id - \beta &= s, \qquad s \in \mathcal{M}_{(n)}.
    \end{aligned}
\end{equation}
We now show that this hierarchy is indeed the conic dual of the sequential NPA hierarchy.

\begin{proposition}\label{prop:SequentialDualSparse}
    The sequential NPA hierarchy \cref{eq:SequentialNPASDP} and the sparse SOS hierarchy \cref{eq:sparseSOSSDP} are conically dual.
\end{proposition}
\begin{proof}
    Define the dual cone of $\mathcal{M}_{(n)}$ as $\mathcal{M}_{(n)}^{\vee} = \{L: V_{(n)} \to \mathds{C} \text{ hermitian linear functional} \mid L(\mathcal{M}_{(n)}) \geq 0 \}$.
    We shall show that every dual feasible solution for the sparse SOS hierarchy corresponds to a feasible moment solution for the sequential NPA hierarchy, and vice versa.
    
    For the easier direction (SOS $\implies$ moment), if $L \in \mathcal{M}_{(n)}^{\vee}$, then by definition for every SOS $f \in \mathcal{M}_{(n)}$ we have $L(f) \geq 0$. We can identify the entries of the moment matrices, analogous to \cref{eq:MomentMatrixStateIdentification}, by
    \begin{align}\label{eq:MomentFunctionalFormula}
        \Theta^{(n)}(a|x)_{w, v} = L(w^* \A v)
    \end{align}
    and check that they satisfy \cref{eq:SequentialNPASDP}.
    
    For the converse (moment $\implies$ SOS), suppose $\{ \Theta^{(n)}(a|x) \}$ is a solution of \cref{eq:SequentialNPASDP}.
    Then for each $a,x$, define linear functionals $L_{a|x}:\cB_{2n}\to\mathds{C}$ from $\Theta^{(n)}(a|x)$.
    Define $L$ on $V_{(n)}$ by
    \begin{align}\label{eq:SparseFunctionalExtension}
        L(\A P):=L_{a|x}(P),\qquad L(P):=\sum_aL_{a|x}(P),\qquad P\in\cB_{2n}.
    \end{align}
    The second expression is independent of $x$ by the strongly no-signaling condition.
    These definitions respect the linear polynomial relations, and $L(\id)=1$.
    The positive semidefiniteness of each $\Theta^{(n)}(a|x)$ implies that, for any polynomial $f\in\cB_n$,
    \begin{align}\label{eq:SparseConditionalSquarePositivity}
        L_{a|x}(f^*f),\,L(f^*f)\geq0.
    \end{align}
    For a fixed input $x$, take $t=\id$ and $f_a,g\in\cB_n$, or $t=B_{b|y},\id-B_{b|y}$ and $f_a,g\in\cB_{n-1}$.
    Using Alice's orthogonality relations and commutation with Bob's letters, we compute
    \begin{equation}\label{eq:SparseFullPositivityCalculation}
        \begin{aligned}
            &L\left(\Bigl(\sum_a\A f_a+g\Bigr)^*t\Bigl(\sum_a\A f_a+g\Bigr)\right)\\
            &\quad=\sum_aL_{a|x}(f_a^*tf_a+f_a^*tg+g^*tf_a)+L(g^*tg)\\
            &\quad=\sum_aL_{a|x}\bigl((f_a+g)^*t(f_a+g)\bigr)\geq0.
        \end{aligned}
    \end{equation}
    The second equality uses $L(g^*tg)=\sum_aL_{a|x}(g^*tg)$, and the final inequality follows from moment-matrix positivity when $t=\id$ and from the corresponding localizing matrix otherwise.
    The completeness terms vanish under $L$ by \cref{eq:SparseFunctionalExtension}.
    It follows that $L$ is nonnegative on the entire $\mathcal{M}_{(n)}$; that is, $L\in\mathcal{M}_{(n)}^{\vee}$.

    Finally, $L(\beta)=\langle\vec{\beta},p\rangle$, so the normalized optimization over $\mathcal{M}_{(n)}^{\vee}$ is the sequential NPA hierarchy, with conic dual \cref{eq:sparseSOSSDP}.
    In particular, a feasible certificate satisfies
    \begin{align}\label{eq:SparseCertificateUpperBound}
        m-\langle\vec{\beta},p\rangle=L(m\id-\beta)=L(s)\geq0.
    \end{align}
\end{proof}

\begin{remark}\label{rem:SparseAlternativeFormulation}
    There is an equivalent formulation of \cref{eq:sparseSOSSDP} that leaves Alice's completeness multipliers explicit.
    Let $\mathcal{M}_{(n)}^0$ be the cone obtained from \cref{eq:SparseSequentialCone} by omitting the terms $\sum_x(\id-\sum_a\A)h_x$.
    We compensate for the smaller cone with Lagrange multipliers $\lambda^x_{u,v}$ for pairs of Bob monomials $u,v\in\cB_n$:
    \begin{equation}\label{eq:sparseSOSAlternativeSDP}
        \begin{aligned}
            \gamevalueSparse{\cG}{n}=\inf_{m\in\mathds{R},\,s,\,\{\lambda^x_{u,v}\}}\quad &m\\
            \text{subject to}\quad &m\id-\beta=s+\sum_x\sum_{u,v}\lambda^x_{u,v}u^*\Bigl(\id-\sum_a\A\Bigr)v,\\
            &s\in\mathcal{M}_{(n)}^0,\qquad \lambda^x_{v,u}=\overline{\lambda^x_{u,v}}.
        \end{aligned}
    \end{equation}
    This alternative formulation satisfies the running intersection property in~\cite{klep2022sparse} and hence belongs to a special case of the sparse SOS optimization.
    Then, it is shown~\cite{klep2022sparse} that, asymptotically as $n \to \infty$, this formulation converges to the standard SOS hierarchy, which is dual to the standard NPA hierarchy.
    At finite levels, however, there is generally no degree guarantee as the sparse SOS certificate generally requires a higher degree than the dense (i.e., usual) SOS hierarchy (analogous to \cref{prop:CompareStandardSequentialNPA}).
    The numerical analysis of sparse SOS hierarchy vs.\ dense SOS hierarchy~\cite[Chapter~6.7]{magron2023sparse} provides insight into the potential numerical performance of the sequential NPA hierarchy vs.\ the standard one due to \cref{prop:SequentialDualSparse}.
\end{remark}

\section*{Author contributions}
The authors are listed in alphabetical order by surname.
MOR and XX conceived the project.
XX developed the original proof strategy, coordinated the project, and led the writing and revision of the manuscript.
All authors contributed to developing and verifying the arguments.
CP, MOR, SS, and LT also contributed to the writing and revisions.
The authors used ChatGPT for English editing in the initial draft and for help during the final revision.

\section*{Acknowledgments}
We thank Matilde Baroni, Dominik Leichtle, Ivan \v{S}upi\'{c}, Thomas Vidick, and Michael Walter for the helpful discussions. In addition, Connor Paddock and Simon Schmidt want to thank Alexander Kulpe, Giulio Malavolta, and Michael Walter for discussions about compilation in the commuting operator framework and its relation to the result $\mathrm{MIP}^{\mathrm{co}}=\mathrm{coRE}$.

MOR, LT and XX acknowledge funding by the ANR for the JCJC grant LINKS (ANR-23-CE47-0003) and T-ERC QNET (ANR-24-ERCS-0008), by INRIA and CIEDS in the Action Exploratoire project DEPARTURE. MOR, IK, LT and XX acknowledge support by the European Union's Horizon 2020 Research and Innovation Programme under QuantERA Grant Agreement no. 731473 and 101017733. 
IK was supported by the Slovenian Research and Innovation Agency program P1-0222 and grants J1-50002, N1-0217, J1-3004,
J1-50001, J1-60011, J1-60025.
Partially supported by the Fondation de l'École polytechnique as part of the Gaspard Monge Visiting Professor Program.
IK thanks École Polytechnique and Inria for hospitality during the preparation of this manuscript. 
SS acknowledges support by the Deutsche Forschungsgemeinschaft (DFG, German Research Foundation) under Germany's Excellence Strategy - EXC 2092 CASA - 39078197. CP acknowledges funding support from the Natural Sciences and Engineering Research Council of Canada (NSERC). YZ is supported by VILLUM FONDEN via QMATH Centre of Excellence grant number 10059 and Villum Young Investigator grant number 37532.

\section*{Note added}
At a late stage of preparing this manuscript, we became aware of independent related results by David Cui, Chirag Falor, Anand Natarajan, and Tina Zhang, which address similar questions.
In particular, they tackle the problem from the sum of squares perspective, which is dual to our NPA moment approach.
We are grateful to them for the coordination prior to submission.

\printbibliography

\end{document}